\documentclass[12pt]{article}
\usepackage{amssymb,calc,enumerate,booktabs}
\usepackage{multirow}
\usepackage{etoolbox,rotating,graphicx}
\usepackage{savesym}
\savesymbol{AND}
\usepackage{algorithmic}
\usepackage{algorithm}
\usepackage{caption}
\usepackage{pdflscape}
\usepackage{tablefootnote}
\usepackage{fixltx2e,url}
\usepackage{savesym}
\usepackage{wrapfig}
\usepackage{amsmath}
\usepackage[subrefformat=parens,labelformat=parens]{subcaption}
\usepackage{comment}
\savesymbol{theoremstyle}
\usepackage{amsthm}
\usepackage{arydshln}
\usepackage{adjustbox}
\usepackage{enumitem}
\usepackage{fixltx2e,url}
\usepackage{physics}
\usepackage{float,xcolor}
\usepackage{pgfplots}
\usepackage{tikz}
\usepackage{booktabs,multirow}
\usepackage[nodayofweek]{datetime}
\newdateformat{mydate}{\twodigit{\THEDAY}{ }\shortmonthname[\THEMONTH], \THEYEAR}
\DeclareGraphicsExtensions{.pdf,.jpeg,.png}
\restoresymbol{theorem}{theoremstyle}
\newtheorem{theorem}{Theorem}

\newtheorem{lemma}[theorem]{Lemma}

\newcommand{\x}{\vb{x}}
\newcommand{\obj}{\operatorname{OBJ}}
\makeatletter
\newsavebox\myboxA
\newsavebox\myboxB
\newlength\mylenA
\newcommand*\myoverline[2][0.75]{%
    \sbox{\myboxA}{$\m@th#2$}%
    \setbox\myboxB\null
    \ht\myboxB=\ht\myboxA%
    \dp\myboxB=\dp\myboxA%
    \wd\myboxB=#1\wd\myboxA
    \sbox\myboxB{$\m@th\overline{\copy\myboxB}$}
    \setlength\mylenA{\the\wd\myboxA}
    \addtolength\mylenA{-\the\wd\myboxB}%
    \ifdim\wd\myboxB<\wd\myboxA%
       \rlap{\hskip 0.5\mylenA\usebox\myboxB}{\usebox\myboxA}%
    \else
        \hskip -0.5\mylenA\rlap{\usebox\myboxA}{\hskip 0.5\mylenA\usebox\myboxB}%
    \fi}
\makeatother
\makeatletter
\newcommand*{\textoverline}[1]{$\myoverline{\hbox{#1}}\m@th$}
\makeatother
\usepackage{relsize}
\newcommand*{\tran}{^{\mkern-1.0mu{\mathsmaller T}}}

 \usepackage[margin=1.0in]{geometry}
 \allowdisplaybreaks[4]
\begin{document}


\title{Revisiting some classical linearizations of the quadratic binary optimization problem}

\author{\sc{Abraham P Punnen}\\
Department of Mathematics, Simon Fraser University,\\ Surrey,
 B C, Canada.{\tt ~~(apunnen@sfu.ca)}\\\AND
\sc{Navpreet Kaur}\\
Department of Mathematics, Simon Fraser University,\\ Surrey,
 B C, Canada.{\tt ~~(navpreet\_kaur\_2@sfu.ca)}}


\maketitle

\begin{abstract}
In this paper, we present several new linearizations of a quadratic binary optimization problem (QBOP), primarily using the method of aggregations. Although aggregations were studied in the past in the context of solving system of Diophantine equations in non-negative variables, none of the approaches developed produced practical models, particularly due to the large size of associate multipliers. Exploiting the special structure of QBOP we show that selective aggregation of constraints provide valid linearizations with interesting properties. For our aggregations, multipliers can be any non-zero real numbers. Moreover, choosing the multipliers appropriately, we demonstrate that the resulting LP relaxations have value identical to the corresponding non-aggregated models. We also provide a review of existing explicit linearizations of QBOP and presents the first systematic study of such models. Theoretical and experimental comparisons of new and existing models are also provided.
\end{abstract}

\section{Introduction}
Let $Q=(q_{ij})$ be an $n\times n$  matrix  and $\vb c\tran=(c_1,c_2,\ldots ,c_n)$ be an $n$-vector. Then, the {\it quadratic binary optimization problem} (QBOP) with linear constraints is to
 \begin{align*}
\text{Maximize}\hspace{11pt} &\sum_{i=1}^n\sum_{j\in R_i}q_{ij}x_ix_j+\sum_{i=1}^nc_ix_i \\
\mbox{Subject to: }~~& \x\in K, \mbox{ and }x_i\in \{0,1\} \mbox{ for all } i=1,2,\ldots ,n,
\end{align*}
where $K$ is a polyhedral set, $\x\tran=(x_1,x_2,\ldots ,x_n)$ and $R_i=\{j : q_{ij}\neq 0\}$. Unless otherwise stated, we assume that the matrix $Q$ is symmetric and its diagonal elements are zeros. Since our focus is primarily on the linearization problem associated with QBOP, for simplicity, we assume that $K$ is the unit cube in $R^n$. i.e. $K=\{0,1\}^n$. However, the theoretical results in this paper extend to any set $K$ that is polyhedral in nature. Consequently, most of the discussions in this paper will be on the {\it quadratic unconstrained binary optimization problem} (QUBO) with the understanding that the corresponding results extend directly to the more general case of QBOP. For a state-of-the-art discussion on QUBO, we refer to the book~\cite{punnen}.

The terminology `linearization' in the context of QBOP generally refers to writing a QBOP as an equivalent {\it mixed integer linear program} (MILP). There is another notion of `linearization' associated with QBOP where conditions are sought on the existence of a row vector $\vb d$ such that $\x^TQ\x = \vb{d}\tran\x$  for all binary $\x\in K$~\cite{hus,kp,punnen3,cela2,adams2,punnen4}. The notion of linearization considered  in this paper is the former.

A linearization, (i.e. an equivalent MILP formulation) of QBOP is said to be {\it explicit} if it contains a variable, say $y_{ij}$, corresponding to the product $x_ix_j$ for each $i\neq j$ such that $y_{ij}=1$ if $x_i=x_j=1$ and zero otherwise. There is another class of linearization studied for QBOP, called {\it compact linearization} that does not use the variables $y_{ij}$~\cite{b2,b3,b1,b9,glover1,b6,b7,b10}. Compact linearizations are investigated by many researchers, particularly in obtaining strong linear programming (LP) relaxation bounds~\cite{b2,b3,b1,b7} that matches the roof duality bound~\cite{ad}. The major advantage of a compact linearization is that it has relatively less number of variables and constraints, compared to explicit linearizations. However,  many such models use numbers of large absolute values in defining bounds on variables or expressions and this could lead to numerical difficulties. In a companion paper, we present a through analysis of existing as well as new compact linearizations~\cite{pk2}.  These formulations provide strong relaxations, but requires to use SDP solvers. Researchers also studied convex reformulations~\cite{b1,qcrb} of QBOP making use of results from semidefinite programming.

Our main focus in this paper is to study explicit formulations in a systematic way both from theoretical as well as computational perspectives. Though traditional explicit linearizations have a large number of constraints and variables, some of them provide strong upper bounds when the integrality constraints are relaxed. For example, the Glover-Woolsey linearization~\cite{gw1} (also called the {\it standard linearization}) have an LP relaxation value that matches other lower bounds such as the roof duality bound~\cite{ad}. Since  explicit formulations use the $y_{ij}$ variables, insights gained from the polyhedral structure of the Boolean quadric Polytope~\cite{pad} and valid inequalities derived from the support graph of the $Q$ matrix~\cite{gueye}  can be readily used to strengthen the corresponding LP relaxations and to aid development of branch and cut algorithms.

We distinguish between two classes of explicit formulations. In one case, the variables $y_{ij}$  defines the product $x_ix_j$ precisely for all feasible solutions $\x$ of the model. i.e. $y_{ij}$ takes zero or one values only and $y_{ij}=1$ precisely when $x_i=x_j=1$. We call MILP formulations having this property as {\it precise models}. In another class of MILP models, the above definition of the $y_{ij}$ variables is guaranteed only when $\x$ is an optimal solution of the model. We call such MILP formulations the {\it optimality restricted models}. Optimality restricted models normally have lesser number of constraints and mostly run faster than precise models. When the model is solved to optimality, the optimal solution and the objective function value produced are indeed accurate under both precise and optimality restricted models. On the other hand, if optimality is not reached and a solver terminates the program due to, say time limit restrictions, the objective function value of the best solution produced, say $\x^0$, may not be equal to $f(\x^0)=(\x^0)\tran Q\x^0 +\textbf{ c}\tran\x^0$ for an optimality restricted model. To get the precise objective function value of such a (heuristic) solution, one may need to recompute $f(\x^0)$.  After recomputing, the resulting objective function value is unpredictable and could be significantly worse or significantly better than the value reported by the solver. Later, we will illustrate this with an example. Thus, it is important to study both precise and optimality restricted models separately.

The MILP formulations are not simply a tool to compute an exact optimal solution. Considering the significant power of modern day MILP solvers such as Gurobi and Cplex, they are often used to obtain heuristic solutions  by imposing time restrictions or could be used in developing sophisticated matheuristics~\cite{mathe}. For such purposes, precise models are more valuable. However, to simply compute an optimal solution, optimality restricted models are, in general,  more useful. Because of this, after discussing each precise model, we also present a corresponding optimality restricted model, whenever applicable.

In this paper, we present the first systematic experimental and theoretical study of explicit linearizations of QUBO. Starting with a review of existing explicit linearizations of QUBO, we present different ways of generating new explicit linearizations and analyze the quality of them from an LP relaxation point of view. One class of our new explicit linearizations uses only $2n$ general constraints which matches the number of general constraints in many of the compact formulations~\cite{b2,b3,b1,b7} and also have a LP relaxation bound that matches the bound of the standard linearization. To the best of our knowledge, no such compact explicit linearizations are known before. We also present detailed and systematic experimental analysis using all of the explicit linearizations we obtained offering insights into the relative merits of these models.

One of the techniques we use in obtaining new linearizations is by making use of selective aggregation of constraints. Weighted aggregation that preserve the set of feasible solutions of a Diophantine system in non-negative variables have been investigated by many researchers~\cite{ag1,ag2,ag3,ag4,ag5,ag6,ag7,ag8}. While theoretically interesting, such an approach never worked well in practice, particularly due to the large coefficients associated with the aggregations and weaker LP relaxation bounds that results in. We show that selective aggregation of constraints provides a valid MILP model for QUBO for any selection of weights in the aggregation. In particular, by choosing these weights carefully, we can achieve an LP relaxation bound that is the same as that of the non-aggregated model. Our work also have linkages with surrogate duality providing strong duality relationship, unlike weak duality generally seen in MILP~\cite{dyer,glover2,glover3,glover4}.
\section{Basic Explicit MILP models }
Let us now discuss four basic explicit linearizations of QUBO and their corresponding optimality restricted versions. Linearizing the product $x_ix_j$ of binary variables is achieved by replacing it with a variable $y_{ij}$ and introduce additional constraints to force $y_{ij}$ to take 0 or 1 values only in such a way that $y_{ij}=1$ if and only if $x_i=x_j=1$. The process of replacing  $x_ix_j$ by a single variable $y_{ij}$ as discussed above is called {\it linearization of the product} $x_ix_j$. The product $x_ix_j$  can be linearized in many ways. Perhaps the oldest such linearization technique was discussed explicitly by Watters~\cite{watters}, which was implicit in other works as well~\cite{dan1,fortet,c1,goldman,zangwill}. To linearize the product $x_ix_j$, Watters~\cite{watters} introduced the constraints (which is implicit in the work of Dantzig~\cite{dan1})
\begin{align}
\label{ch5-e1}&x_i+x_j -y_{ij}\leq 1 \\
\label{ch5-e2}&2y_{ij}-x_i-x_j\leq 0\\
\label{ch5-e3}&y_{ij},x_i,x_j\in \{0,1\}
\end{align}
This leads to the following $0-1$ linear programming formulation of QUBO~\cite{dan1,watters}.
 \begin{align}
\nonumber\text{DW:\quad Maximize}\hspace{11pt} &\sum_{i=1}^n\sum_{j\in R_i}q_{ij}y_{ij}+\sum_{i=1}^nc_ix_i \\
\label{pe1}\mbox{Subject to: }~~&  x_i+x_j -y_{ij}\leq 1 \mbox{ for all } j\in R_i, i=1,2,\ldots ,n,\\
\label{pe2}&  2y_{ij}-x_i-x_j\leq 0 \mbox{ for all } j\in R_i, i=1,2,\ldots ,n,\\
\label{pe3}&x_i\in \{0,1\} \mbox{ for all } i=1,2,\ldots ,n,\\
\label{pe4}&y_{ij} \in \{0,1\} \mbox{ for all } j\in R_i,i=1,2,\ldots, n.
\end{align}
This formulation has at most $2n(n-1)$ general constraints and $n^2+n$ binary variables.

Using a disaggregation of constraints~\eqref{ch5-e2}, Glover and Woolsey~\cite{gw} proposed  to linearize the product $x_ix_j$ by introducing the constraints
\begin{align}
\label{ch5e0}&x_i+x_j -y_{ij}\leq 1 \\
\label{ch5-e4}&y_{ij} \leq x_i\\
\label{ch5-e5}&y_{ij}\leq x_j\\
\label{ch5-e6}&y_{ij}\geq 0, x_i,x_j\in\{0,1\}.
\end{align}
Note that $y_{ij}$ is now a continuous variable and the number of constraints for each $(i,j)$ pair is increased by one. This leads to the following $0-1$ programming formulation of QUBO~\cite{gw}.
 \begin{align}
\nonumber\text{GW:\quad Maximize}\hspace{11pt} &\sum_{i=1}^n\sum_{j\in R_i}q_{ij}y_{ij}+\sum_{i=1}^nc_ix_i \\
\label{ch5-e7}\mbox{Subject to: }~~& x_i+x_j-y_{ij} \leq 1 \mbox{ for all } j\in R_i, i=1,2,\ldots ,n,\\
\label{ch5-e9}&  y_{ij}-x_i \leq 0 \mbox{ for all } j\in R_i, i=1,2,\ldots ,n\\
\label{ch5-e8}& y_{ji}-x_i \leq 0 \mbox{ for all } j\in S_i, i=1,2, \ldots ,n\\
\label{ch5-e10}&x_i\in \{0,1\} \mbox{ for all } i=1,2\ldots ,n.\\
\label{ch5-e11}&y_{ij} \geq 0 \mbox{ for all } j\in R_i, i=1,2,\ldots ,n.
\end{align}
where $S_i=\{j: q_{ji}\neq 0\}$. Since $Q$ is symmetric, $R_i=S_i$ for $i=1,2,\ldots ,n$. The formulation GW has $3n(n-1)$ general constraints, an increase compared to DW. GW also has $n$ binary variables and $n(n-1)$ continuous variables. Note that the constraints \eqref{ch5-e9} and \eqref{ch5-e8} are the restatements of the constraints \eqref{ch5-e4} and \eqref{ch5-e5} when taken over all the pairs $i$ and $j$. We use this representation to aid easy presentation of some of our forthcoming results. In the literature the constraints \eqref{ch5-e4} and \eqref{ch5-e5} are normally used instead of \eqref{ch5-e9} and \eqref{ch5-e8} when describing GW.  GW is perhaps the most popular linearization studied in the literature for QUBO and  is often called the {\it standard linearization} or {\it Glover-Woolsey linearization}.

A variation of GW, often attributed to the works of Fortet~\cite{fortet,c1} can be stated as follows.
 \begin{align}
\nonumber\text{FT:\quad Maximize}\hspace{11pt} &\sum_{i=1}^n\sum_{j\in R_i}q_{ij}y_{ij}+\sum_{i=1}^nc_ix_i \\
\label{f4}\mbox{Subject to: }~~& x_i+x_j-y_{ij} \leq 1 \mbox{ for all } j\in R_i, i=1,2,\ldots ,n,\\
\label{f5}&  y_{ij}-x_i \leq 0 \mbox{ for all } j\in R_i, i=1,2,\ldots ,n\\
\label{f6}& y_{ij}-y_{ji} \leq 0 \mbox{ for all } j\in R_i, i=1,2, \ldots ,n\\
\label{f7}&x_i\in \{0,1\} \mbox{ for all } i=1,2\ldots ,n.\\
\label{f8}&y_{ij} \geq 0 \mbox{ for all } j\in R_i, i=1,2,\ldots ,n.
\end{align}
The constraints \eqref{f6} implies $y_{ij}=y_{ji}, j\in R_i, i=1,2,\ldots ,n.$ In fact, we can replace constraints \eqref{f6} by
\begin{equation}\label{ftem-n}
y_{ij}=y_{ji} \mbox{ for all } j\in R_i, j > i, i=1,2, \ldots ,n.
\end{equation}
The formulation FT has $3n(n-1)$ general constraints, $n$ binary variables, and $n(n-1)$ continuous variables. The number of constraints reduces to $\frac{5}{2}n(n-1)$ when constraints \eqref{f6} is replaced by \eqref{ftem-n}.

Let us now propose a new basic linearization of QUBO. Note that DW can be viewed as pairwise aggregation of the constraints \eqref{ch5-e4} and \eqref{ch5-e5} of GW along with additional binary restriction on the $y_{ij}$ variables. Using this line of reasoning and applying pairwise aggregation of  constraints \eqref{ch5-e9} and \eqref{ch5-e8} of GW, we get the  linearization
 \begin{align}
\nonumber\text{PK:\quad Maximize}\hspace{11pt} &\sum_{i=1}^n\sum_{j\in R_i}q_{ij}y_{ij}+\sum_{i=1}^nc_ix_i \\
\label{pe5}\mbox{Subject to: }~~&  x_i+x_j -y_{ij}\leq 1 \mbox{ for all } j\in R_i, i=1,2,\ldots ,n,\\
\label{pe6}& y_{ij}+y_{ji}-2x_i\leq 0 \mbox{ for all } j\in R_i, i=1,2,\ldots ,n,\\
\label{pe7}&x_i\in \{0,1\} \mbox{ for all } i=1,2\ldots ,n.\\
\label{pe8}& y_{ij} \geq 0 \mbox{ for all }  j\in R_i,i=1,2,\ldots, n.
\end{align}
\begin{theorem}
PK is a valid MILP formulation of QUBO.
\end{theorem}
\begin{proof}
Consider the variables $x_i$ and $x_j$. When $x_i=0$ or $x_j=0$, by constraint~\eqref{pe6}, $y_{ij}=y_{ji}=0$. Suppose $x_i=x_j=1$. Then, from constraint \eqref{pe5},
\begin{equation}\label{pea}
y_{ij} \geq 1\mbox{ and }y_{ji}\geq 1.
\end{equation} From constraint \eqref{pe6}, we have $y_{ij}+y_{ji}\leq 2$. In view of \eqref{pea}, $y_{ij}=y_{ji}=1$.
\end{proof}
PK has at most $2n(n-1)$ general constraints, $n$ binary variables, and $n(n-1)$ continuous variables. Unlike DW, the variables $y_{ij}$ in PK are continuous and it is  the most compact basic linearization of QUBO.
\subsection{Optimality restricted variations of the basic models}
In Section 1, we discussed the concept of optimality restricted variations of a linearization of QUBO, their merits and drawbacks, and why is it important and relevant to investigate them separately from exact models. Let us now discuss the optimality restricted variations derived from the four basic linearization models discussed in this section earlier. Let $R^+_i=\{j : q_{ij}> 0\}$ and $R^-_i=\{j: q_{ij}< 0\}$. First note that for DW, the binary restriction of $y_{ij}$ variables can be relaxed to $y_{ij}\geq 0$ for $j\in R^-_i, i=1,2,\ldots ,n$. This provides an optimality restricted variation of DW. To see why this works, note that when $x_i=x_j=0$, $y_{ij}=0$, without requiring $y_{ij}$ to be binary. Likewise, when $x_i=x_j=1$, $y_{ij}=1$, without requiring $y_{ij}$ to be binary. The binary restriction is relevant only when precisely one of $x_i$ or $x_j$ is 1. In this case constraint \eqref{pe2} becomes $y_{ij}\leq \frac{1}{2}$. The binary restriction then force $y_{ij}$ to be zero. But, in this case, if $q_{ij} < 0$ then $y_{ij}$ will take value $0$ at optimality regardless, $y_{ij}$ is binary or not. In fact, we can simplify the model further.
Note that the general linearization constraints \eqref{ch5-e1} and \eqref{ch5-e2} can be rewritten as
\begin{equation*}
x_i+x_j-1 \leq y_{ij} \leq (1/2)(x_i+x_j)
\end{equation*}
 Thus, at the optimality, constraint \eqref{pe1} is not active if $j\in R^+_i$ and hence can be removed. Similarly, constraint \eqref{pe2} can be removed if $j\in R^-_i$. Thus, an optimality restricted version of DW can be stated as
 \begin{align}
\nonumber\text{ORDW:\quad Maximize}\hspace{11pt} &\sum_{i=1}^n\sum_{j\in R_i}q_{ij}y_{ij}+\sum_{i=1}^nc_ix_i \\
\label{pe1x}\mbox{Subject to: }~~&  x_i+x_j -y_{ij}\leq 1 \mbox{ for all } j\in R^-_i, i=1,2,\ldots ,n,\\
\label{pe2x}&  2y_{ij}-x_i-x_j\leq 0 \mbox{ for all } j\in R^+_i, i=1,2,\ldots ,n,\\
\label{pe3x}&x_i\in \{0,1\} \mbox{ for all } i=1,2\ldots ,n,\\
\label{pe4x}&y_{ij} \in \{0,1\} \mbox{ for all } j\in R_i^+,i=1,2,\ldots, n,\\
\label{pe4xx}&y_{ij} \geq 0 \mbox{ for all } j\in R_i^-,i=1,2,\ldots, n.
\end{align}
The formulation ORDW has $\sum_{i=1}^{n}\left(|R^{-}_{i}|+|R^{+}_{i}|\right)$ general constraints and $\left(n+\sum_{i=1}^{n}|R^{-}_{i}|\right)$ binary variables.

While the optimality restricted model ORDW guarantees an optimal solution, if the solver is terminated before reaching an optimal solution, the $y_{ij}$ value from the resulting best solution found need not be equal to $x_ix_j$. Consequently, the reported objective function value could be erroneous and the correct QUBO objective function value needs to be recomputed using the available values of $x_i,i=1,2,\ldots ,n$.  Let us illustrate this with an example.  Consider the instance EX1 of QUBO with $$Q=\begin{bmatrix}0&\alpha\\ \alpha &0\end{bmatrix}\mbox{ and }\vb c\tran=\begin{bmatrix}0&0\end{bmatrix},$$
where $\alpha > 0$.
Suppose that we use ORDW to solve EX1 and terminate the solver at the solution $x_1=x_2=1,y_{12}=y_{21}=0$. Note that this is a feasible solution to ORDW (but not feasible for DW) with objective function value 0. But the QUBO objective function value of this solution is $2\alpha$, which is arbitrarily larger than the value solver produced for the model ORDW when terminated prematurely.

Thus, when MILP models in optimality restricted form are used, it is important to recompute the objective function value in order   to use the model as a heuristic by choosing the best solution produced within a given time/space limit. For precise models like DW, such an issue does not arise.

Let us now look at an optimality restricted version of GW.
Note that the linearization constraints in GW can be restated as
\begin{equation*}
\max\{0,x_i+x_j-1\}\leq y_{ij}\leq \min\{x_i,x_j\}
\end{equation*}
Thus, constraints \eqref{ch5-e7} and \eqref{ch5-e11} are redundant at optimality if $j\in R^+_i$ and constraints \eqref{ch5-e9} and \eqref{ch5-e8} are redundant at optimality if $j\in R^-_i$. Removing these redundant constraints, GW can be restated in the optimality restricted form as
 \begin{align}
\nonumber\text{ORGW:\quad Maximize}\hspace{11pt} &\sum_{i=1}^n\sum_{j\in R_i}q_{ij}y_{ij}+\sum_{i=1}^nc_ix_i \\
\label{ch5-e7x}\mbox{Subject to: }~~& x_i+x_j-y_{ij} \leq 1 \mbox{ for all } j\in R^-_i, i=1,2,\ldots ,n,\\
\label{ch5-e9x}&  y_{ij}-x_i \leq 0 \mbox{ for all } j\in R^+_i, i=1,2,\ldots ,n\\
\label{ch5-e8x}& y_{ji}-x_i \leq 0 \mbox{ for all } j\in S^+_i, i=1,2, \ldots ,n\\
\label{ch5-e10x}&x_i\in \{0,1\} \mbox{ for all } i=1,2\ldots ,n,\\
\label{ch5-e11x}&y_{ij} \geq 0 \mbox{ for all } j\in R^-_i, i=1,2,\ldots ,n.
\end{align}
The formulation ORGW has at most $2n(n-1)$ general constraints. Precisely, the number of general constraints will be $\sum_{i=1}^{n}\left(2 |R_{i}^{+}|+|R_{i}^{-}|\right)$

The main linearization constraints of FT can be stated as
\begin{equation*}
\max\{0,x_i+x_j-1\}\leq y_{ij}\leq x_i
\end{equation*}
Thus, constraints \eqref{f4} and \eqref{f8} are redundant at optimality if $j\in R^-_i$ and \eqref{f5} and \eqref{f6} are redundant at optimality if $j\in R^+_i$. Removing these redundant constraints, FT can be restated in a modified form as
 \begin{align}
\nonumber\text{ORFT:\quad Maximize}\hspace{11pt} &\sum_{i=1}^n\sum_{j\in R_i}q_{ij}y_{ij}+\sum_{i=1}^nc_ix_i \\
\label{f4x}\mbox{Subject to: }~~& x_i+x_j-y_{ij} \leq 1 \mbox{ for all } j\in R^-_i, i=1,2,\ldots ,n,\\
\label{f5x}&  y_{ij}-x_i \leq 0 \mbox{ for all } j\in R^+_i, i=1,2,\ldots ,n,\\
\label{f6x}& y_{ij}\leq y_{ji} \mbox{ for all } j\in R^+_i,  i=1,2, \ldots ,n,\\
\label{f7x}&x_i\in \{0,1\} \mbox{ for all } i=1,2\ldots ,n,\\
\label{f8x}&y_{ij} \geq 0 \mbox{ for all } j\in R^-_i, i=1,2,\ldots ,n.
\end{align}
The model ORFT has $\left(\sum_{i=1}^{n}(2 |R_{i}^{+}|+|R_{i}^{-}|\right)$ general constraints, $n$ binary variables, and $\left(\sum_{i=1}^{n}( |R_{i}^{-}|\right)$ continuous variables.
%

The linearization constraint \eqref{pe5} of PK can be written as \begin{equation*}
\max\{0,x_i+x_j-1\}\leq y_{ij}
\end{equation*}
Thus, the constraints \eqref{pe5} and \eqref{pe8} are redundant at optimality for $j\in R^+_i$. For $q_{ij}=q_{ji}>0$, if at least one of $x_{i}\mbox{ or }x_{j}=0$, from the optimality sense, $y_{ij}=y_{ji}=0$, but if $x_{i}=x_{j}=1$, then the  other linearization constraint
\begin{equation*}
y_{ij}+y_{ji}\leq 2x_i
\end{equation*}
 will become $y_{ij}+y_{ji}\leq 2$ and will be satisfied with equality by an optimal solution in this case. Thus, here upper bound restrictions on $y_{ij}$ are not required. Since $Q$ is symmetric, $q_{ij}=q_{ji}$ and hence they have the same sign. As a result,
 constraint \eqref{pe6} is redundant at optimality for $j\in R^-$. To avoid unboundedness, we need to add the constraint $y_{ij} \geq 0 \mbox{ for all }  j\in R^-_i,i=1,2,\ldots, n$. Thus, the optimality restricted version of PK can be stated as
 \begin{align}
\nonumber\text{ORPK:\quad Maximize}\hspace{11pt} &\sum_{i=1}^n\sum_{j\in R_i}q_{ij}y_{ij}+\sum_{i=1}^nc_ix_i \\
\label{pe5a}\mbox{Subject to: }~~&  x_i+x_j -y_{ij}\leq 1 \mbox{ for all } j\in R^-_i, i=1,2,\ldots ,n,\\
\label{pe6a}& y_{ij}+y_{ji}-2x_i\leq 0 \mbox{ for all } j\in R^+_i, i=1,2,\ldots ,n,\\
\label{pe7a}&x_i\in \{0,1\} \mbox{ for all } i=1,2\ldots ,n.\\
\label{pe8a-1}& y_{ij} \geq 0 \mbox{ for all }  j\in R^-_i,i=1,2,\ldots, n.
\end{align}
The number of general constraints and continuous variables in ORPK is $\sum_{i=1}^{n}\left(|R^{-}_{i}|+|R^{+}_{i}|\right)$ and $\left(\sum_{i=1}^{n}|R^{-}_{i}|\right)$  respectively.

\section{Linearization by Weighted Aggregation}\label{s2}
Let us now discuss some general techniques to generate new explicit linearizations of QUBO from DW, GW, FT and PK using weighted aggregation of selected constraints and thereby reducing the total number of constraints in the resulting model. In general, for integer (binary) programs, weighted aggregation of constraints alter the problem characteristics and the resulting models only provide an upper bound. This is the principle used in surrogate constraints and related duality~\cite{dyer,glover2,glover3,glover4}. Precise aggregations are also studied in the literature but are not suitable for practical applications. Interestingly, by selecting constraints to be aggregated carefully, we show that new precise formulations for QUBO can be obtained with interesting theoretical properties and have practical value.

An unweighted aggregation model derived from GW was presented by Glover and Woolsey~\cite{gw1} using constraints \eqref{ch5-e9} and \eqref{ch5-e8} with amendments provided by Goldman~\cite{goldman}. To the best of our knowledge, no other aggregation  models (weighted or unweighted) are discussed in the literature for explicit linearizations of QUBO using general weights. The effect of these aggregations on the quality of the corresponding LP relaxations will be discussed later.

\subsection{Weighted aggregation of type 1 linearization constraints}

Recall that all of the four basic models have the constraints
$$x_i+x_j-y_{ij} \leq 1 \mbox{ for all } j\in R_i, i=1,2,\ldots ,n$$
and we call them {\it type 1 linearization constraints}.
Let $\alpha_{ij} > 0, j\in R_i, i=1,2,\ldots ,n$. Consider the system of $n$ linear inequalities
\begin{equation}\label{pe9}
\sum_{j\in R_i}\alpha_{ij}(x_j-y_{ij})\leq \left(\sum_{j\in R_i}\alpha_{ij}\right)(1-x_i), \mbox{ for } i=1,2,\ldots ,n.
\end{equation}
obtained by weighted aggregation of type 1  linearization constraints.

Let DW($\alpha$)  be the mixed integer linear programs obtained from DW by replacing constraints set  \eqref{pe1} by \eqref{pe9}. When $q_{ij}=0$, both $y_{ij}$ and $y_{ji}$ are not defined. If $q_{ij}\neq 0$ by the symmetry of $Q$, $q_{ji}=q_{ij}\neq 0$. Thus, $y_{ij}$ is defined if and only if $y_{ji}$ is defined. Moreover, when both $y_{ij}$ and $y_{ji}$ are defined,  $j\in R_i$ if and only if $i\in R_j$.

\begin{theorem}\label{th1-dw}
  DW($\alpha$) is a valid formulation of QUBO.
\end{theorem}
%
\begin{proof}
For any pair $(i,j)$ with $j\in R_i$ we have to show that $y_{ij}=y_{ji}=x_ix_j$. Suppose $x_i=0$. Choose any $j\in R_i$. By constraint set \eqref{pe2} we have $$2y_{ij}\leq x_j \mbox{ and } 2y_{ji}\leq x_j$$
Since $x_{ij}\in \{0,1\}$ and $y_{ij},y_{ji}\in \{0,1\}$ we have $y_{ij}=0$ and $x_{ij}=0$. Now, suppose $x_i=1$. Choose any $k\in R_i$. From the constraint set \eqref{pe9},
$$\sum_{j\in R_i}\alpha_{ij}(x_j-y_{ij})\leq 0$$
If $x_j-y_{ij} < 0$ for some $j\in R_i$, then $y_{ij}\neq 0$. Since $y_{ij}\in \{0,1\}$, $y_{ij}=1$ and hence $x_j=0$. But we already established that if $x_j=0$ then $y_{ij}=0$, a contradiction. If $x_j-y_{ij} > 0$ then there exists a $k\in R_i$ such that $x_k-y_{ik} < 0$ which is impossible as established above. Thus $x_j-y_{ij}=0$ for all $j\in R_i$. Thus, $y_{ij}=x_j$ for all $j\in R_i$. Thus, for any $j\in R_i$ if $x_j=1$ then $y_{ij}=1$. Since $j\in R_i$ we have $i\in R_j$. Thus, if $x_j=1$ and $x_i=1$ following similar analysis we get $y_{ji}=1$. This completes the proof.
\end{proof}

Let PK($\alpha$)  be the mixed integer linear programs obtained from PK by  replacing constraints set  \eqref{pe5} with \eqref{pe9} and by adding the upper bound constraints $y_{ij}\leq 1$ for $j\in R_{i},  i=1,2,\ldots ,n.$

\begin{theorem}\label{th1-pk}
PK($\alpha$) is a valid formulation of QUBO.
\end{theorem}
\begin{proof}
If $x_i=0$ for any $i$, then by constraints \eqref{pe6} $y_{ij}+y_{ji}\leq 0$ yielding $y_{ij}=y_{ij}=0$ for any $j$. Now suppose $x_i=1$. Then, from constraints \eqref{pe9}
$$\sum_{j\in R_i}\alpha_{ij}(x_j-y_{ij})\leq 0$$
If $x_j -y_{ij} < 0$ then $y_{ij}\neq 0$. If $0 \leq y_{ij} \leq 1$, then $x_j=0$ which implies $y_{ij}=0$, a contradiction. Thus $y_{ij} > 1$ and this contradicts the upper bound restriction on $y_{ij}$. If $x_j-y_{ij} > 0$ then there exists a $k\in R_i$ such that $x_k-y_{ik} < 0$, a contradiction. Thus, $x_j=y_{ij}$ for all $j\in R_i$. Thus, if $x_j=1$ then $y_{ij}=1$. If $j\in R_i$ then $i\in R_j$. Using similar arguments with $x_j=1$ we get $ y_{ji}=1$. Thus, if $x_i=x_j=1$ then $y_{ij}=y_{ji}=1$ and this concludes the proof.
\end{proof}
The upper bound of 1 on the variables $y_{ij}$ is required for this formulation to be valid. Consider the instance EX2 of QUBO with
 $$Q=\begin{pmatrix}0&3&-6&-3\\ 3&0&6&3 \\-6&6&0&-6 \\ -3&3&-6&0\\ \end{pmatrix} \mbox{ and } \vb c\tran=(-3,-6,0,3)$$\\
choose $\alpha_{14}=3,\; \alpha_{34}=6$, all other $\alpha_{ij}=1, j\in R_{i}$ for $i=1,\dots ,n$. If the constraints $y_{ij}\leq 1$ is not used, an optimal solution obtained by the formulation PK($\alpha$) is
$x_{1}=0,\; x_{2}=x_{3}=x_{4}=1$,\; $y_{23}=y_{42}=2, \;y_{34}=7/6$, rest of the $y_{ij}\mbox{'s are } 0$ with objective function value 8. But the optimal objective function value of this problem is 6 with $x_{2}=x_{3}=1$ and $x_{1}=x_{4}=0$ as optimal solution.

Let GW($\alpha$)  be the mixed integer linear program obtained from  GW by  replacing constraints set  \eqref{ch5-e7} by \eqref{pe9}.
\begin{theorem}\label{th1-gw}
  GW($\alpha$) is a valid formulation of QUBO.
\end{theorem}
\begin{proof}
If $x_i=0$ then by constraints \eqref{ch5-e9} and \eqref{ch5-e8} we have $y_{ij}=y_{ji}=0$ for all $j\in R_i$. Also, we have an implicit upper bound of $1$ on $y_{ij}$ for all $i$ and $j$. Now, following arguments similar to that of proof of Theorem~\ref{th1-pk}, we have if $x_i=x_j=1$ then $y_{ij}=y_{ji}=1$. This completes the proof.
\end{proof}

Finally, let FT($\alpha$)  be the mixed integer linear programs obtained respectively from  FT by  replacing constraints set  \eqref{f4} by \eqref{pe9}.

\begin{theorem}\label{th1-ft}
 FT($\alpha$) is a valid formulation of QUBO.
\end{theorem}
\begin{proof}
If $x_i=0$ then by constraints \eqref{f5} we have $y_{ij}=0$ for all $j\in R_i$ and by constraints \eqref{f6} we have $y_{ji}=0$ for $j\in R_i$. Also, we have an implicit upper bound of $1$ on $y_{ij}$ for all $i$ and $j$. Now, following arguments similar to that of proof of Theorem~\ref{th1-pk}, we have if $x_i=x_j=1$ then $y_{ij}=y_{ji}=1$. This completes the proof.
\end{proof}

When $\alpha_{ij}=1$ for all $i,j$, we have the unit-weight special cases of   DW($\alpha$),GW($\alpha$), FT($\alpha$) and PK($\alpha$).  In this case, the aggregated constraint is
\begin{equation}\label{pe11}
\sum_{j\in R_i}(x_j-y_{ij})\leq |R_i|(1-x_i), \mbox{ for } i=1,2,\ldots ,n.
\end{equation}
Note that the number of general constraints  is $n^2$ in each of the models  DW($\alpha$) and PK($\alpha$), a reduction from $2n(n-1)$ for  DW and PK. The number of general constraints in FT($\alpha$) and GW($\alpha$) is $n(2n-1)$ a reduction from $3n(n-1)$ constraints in FT and GW.

\subsection{Weighted aggregation of type 2 linearization constraints}

The collection of linearization constraints in our basic models that are not of type 1 are called {\it type 2 linearization constraints}. For DW and PK, there is only one class of type 2 linearization constraints. For GW and FT there are two classes of type 2 linearization constraints. In such cases, there are different types of aggregations possible.

Let us start with aggregation of type 2 linearization constraints \eqref{pe6} from PK.
For $\beta_{ij} > 0, j\in R_i, i=1,2,\ldots ,n$, consider the inequality\\
\begin{equation}\label{pe12}
\sum_{j\in R_i}\beta_{ij}(y_{ij}+y_{ji})\leq 2\left(\sum_{j\in R_i}\beta_{ij}\right)x_i, \mbox{ for } i=1,2,\ldots ,n.
\end{equation}
Let PK($*,\beta$) be the MILP obtained from PK by replacing \eqref{pe6} with \eqref{pe12} and adding the constraints $y_{ij}\leq 1, j\in R_i, i=1,2,\ldots ,n$. The notation PK($*,\beta$) indicates type 1 linearization constraints are not aggregated and type 2 linearization constraints are aggregated using parameter values represented by $\beta_{ij}$.
\begin{theorem}
PK($*,\beta$) is a valid MILP model for QUBO.
\end{theorem}
\begin{proof}
If $x_i=0$ then constraint \eqref{pe12} implies $y_{ij}=-y_{ji}, j\in R_i$ which guarantees that $y_{ij}=y_{ji}=0$. If $x_i=1$ then constraint \eqref{pe12} is redundant. But then, constraint \eqref{pe5} assures that $y_{ij}\geq x_j$ and $y_{ji}\geq x_j$. We also have the constraint $y_{ij}\leq 1$. Thus, if $x_j$ and $x_i$ are both 1, then $y_{ij}=1$ and $y_{ji}=1$. If $x_i=1$ but $x_j=0$ then  constraint  \eqref{pe12} with $x_j$ on the RHS insures that $y_{ji}=y_{ij}=0$ and the result follows.
\end{proof}
For the special case of PK($*,\beta$) where $\beta_{ij}=1$ for all $j\in R_i, i=1,2,\ldots ,n$, the constraint \eqref{pe12} becomes
\begin{equation}\label{pe13}
\sum_{j\in R_i}(y_{ij}+y_{ji})\leq 2|R_i|x_i, \mbox{ for } i=1,2,\ldots ,n.
\end{equation}
PK($*,\beta$) have at most $n^2$ general constraints, a reduction from $2n(n-1)$  for PK.

Although the weighted aggregation of type 2 linearization constraints produced the valid MILP model PK($*,\beta$) for QUBO, weighted aggregation of type 2 linearization constraints of DW  need not produced a valid MILP model for QUBO. Let $\beta_{ij} > 0, j\in R_i, i=1,2,\ldots ,n$. Consider
\begin{equation}\label{pe13dw}
\sum_{j\in R_i}(2y_{ij}-x_j)\leq \left(\sum_{j\in R_i}\beta_{ij}\right)x_i, \mbox{ for } i=1,2,\ldots ,n.
\end{equation}
Let DW($*,\beta$) be MILP obtained from DW replacing constraints \eqref{pe2} by \eqref{pe13dw}. Then, DW($*,\beta$) is not a valid MILP model for QUBO. For example, consider the instance(EX3) of QUBO where $$Q=\begin{pmatrix}0&1&1\\ 1&0&0 \\1&0&0 \\ \end{pmatrix} \mbox{ and } \textbf{c}\tran=(1,-5,-5)$$\\
An optimal solution to the formulation DW($*,\beta$) obtained from this instance is $x_{1}=1,\; x_{2}=x_{3}=0,\; y_{12}=1,\;y_{23}=y_{31}=y_{13}=y_{21}=y_{32}=0$ with objective function value 2 where $\beta_{ij}=1,\;\forall\; j\in R_{i},\; i=1,\ldots,n$. Note that, $y_{12}\neq x_{1}x_{2}$. The optimal objective optimal function value for this problem is 1.

Weighted aggregation of type 2 linearization constraints \eqref{f5} in FT yields another MILP formulation of QUBO. Let $\gamma_{ij} > 0, j\in R_i, i=1,2,\ldots ,n$. Now consider the inequality
\begin{equation}\label{f9}
\sum_{j\in R_i}\gamma_{ij}y_{ij}\leq \left(\sum_{j\in R_i}\gamma_{ij}\right)x_i, \mbox{ for } i=1,2,\ldots ,n.
\end{equation}
Let FT($*,\gamma$) be the MILP obtained from FT by replacing the inequality \eqref{f5} in FT with \eqref{f9} and adding an upper bound of 1 on all $y_{ij}$ variables.
\begin{theorem}
FT($*,\gamma$) is a valid MILP formulation of QUBO.
\end{theorem}
\begin{proof}
When $x_i=0$, constraint \eqref{f9} assures that $y_{ij}=0$. When $x_i=1$ constraint \eqref{f9} is redundant. However,  in this case, from constraints \eqref{f4} and the upper bound restrictions on $y_{ij}$ we have $1\geq y_{ij}\geq x_j$. Thus, if $x_j=1$ then $y_{ij}=1$. If $x_i=1$ but $x_j=0$ then the constraint \eqref{f9} corresponding to $x_j$ assures that $y_{ji}=0$. But then, by constraint \eqref{f6}, $y_{ij}=0$ and the result follows.
\end{proof}
FT($*,\gamma$) have at most $n(2n-1)$ general constraints and $n(n-1)$ upper-bound constraints.
We can also use weighted aggregation of  constraints \eqref{f6} in FT to yield a valid MILP model for QUBO. Let $\theta_{ij} > 0, j\in R_i, i=1,2,\ldots ,n$. Now consider the inequality
\begin{equation}\label{fte1x-n}
\sum_{j\in R_i}\theta_{ij}(y_{ij}-y_{ji})\leq0 \mbox{ for } i=1,2,\ldots ,n.
\end{equation}

Let FT($*,\theta$)  be the MILP obtained from FT by replacing constraints \eqref{f6} with \eqref{fte1x-n} and FT($*, \gamma,\theta$)be the MILP obtained by replacing constraints  \eqref{f5} by \eqref{f9} and \eqref{f6} by \eqref{fte1x-n} along with the addition of $y_{ij}\leq 1\;\forall j\in R_i, i=1,\ldots,n$.
\begin{theorem}\label{t1}
FT($*,\theta$) and FT($*, \gamma,\theta$) are valid MILP formulations of QUBO.
\end{theorem}
We skip the proof of this theorem as it is easy to construct.

FT($*, \gamma,\theta$) has $n^2+n$ general constraints and both FT($*,\theta$) and FT($*, \gamma$) have $n(2n-1)$ general constraints.  The corresponding unweighted versions of these formulations can be obtained by using $\gamma_{ij}, \;\theta_{ij}\mbox{ and }\delta_{ij}=1,\; \forall j\in R_i, i=1,\ldots,n$.\\
Now, we will explore the weighted aggregation of the type 2 linearization constraints \eqref{ch5-e9} and \eqref{ch5-e8} of GW. Let  $\gamma_{ij} > 0$ and $\delta_{ij} > 0$ for all $j\in R_i, i=1,2,\ldots n$ (note that since $Q$ is assumed to be symmetric, $R_i=S_i$). Now, consider the inequality
\begin{equation}\label{pe14}
\sum_{j\in R_i}\gamma_{ij}y_{ij}+\sum_{j\in S_i}\delta_{ij}y_{ji}\leq \left(\sum_{j\in R_i}\gamma_{ij}+\sum_{j\in S_i}\delta_{ij}\right)x_i, \mbox{ for } i=1,2,\ldots ,n.
\end{equation}
Let GW($*,\gamma+\delta$) be the MILP obtained from GW by replacing \eqref{ch5-e9} and \eqref{ch5-e8} in GW by \eqref{pe14} and adding the upper bound constraints $y_{ij} \leq 1$ for all $i$ and $j$.
\begin{theorem}
GW($*,\gamma+\delta$) is a valid MILP model for QUBO.
\end{theorem}
\begin{proof}
When $x_i=0$, constraint \eqref{pe14} assures that $y_{ij}=y_{ji}=0$. When $x_i=1$ constraint \eqref{pe14} is redundant. However,  in this case, from constraints \eqref{ch5-e7} and the upper bound restrictions on $y_{ij}$ we have $1\geq y_{ij}\geq x_j$ and $1\geq y_{ji}\geq x_j$. Thus, if $x_j=1$ we have $y_{ij}=y_{ji}=1$. If $x_j=0$ and $x_i=1$, using constraint  \eqref{pe14} corresponding to $x_j$, we have $y_{ji}=y_{ij}=0$ and the result follows.
\end{proof}
Unlike GW($*,\gamma+\delta$), the aggregation of \eqref{{fte1x-n}} and the following constraint
\begin{equation}\label{fte1x-n2}
\sum_{j\in R_i}\theta_{ij}(y_{ij}-y_{ji})\leq0 \mbox{ for } i=1,2,\ldots ,n.
\end{equation}
need not produce a valid MILP model corresponding to FT.

Note that PK($*,\beta$) is the same as GW($*,\beta+\beta$). When $\gamma_{ij}=\delta_{ij}=1$ for all $i$ and $j$, GW($*,\gamma+\delta$) becomes the {\it Glover-Woolsey aggregation} discussed in~\cite{gw1} along with the amendments proposed by Goldman~\cite{goldman}.

Let us now consider weighted aggregation of constraints \eqref{ch5-e9} and \eqref{ch5-e8} of GW separately.  Consider the inequalities
\begin{align}\label{pe21}
\sum_{j\in R_i}\gamma_{ij}y_{ij}&\leq \left(\sum_{j\in R_i}\gamma_{ij}\right)x_i, \mbox{ for } i=1,2,\ldots ,n.\\
\label{pe22}\sum_{j\in S_i}\delta_{ij}y_{ji}&\leq \left(\sum_{j\in S_i}\delta_{ij}\right)x_i, \mbox{ for } i=1,2,\ldots ,n.
\end{align}
In GW, replace constraints \eqref{ch5-e9} and \eqref{ch5-e8} by \eqref{pe21} and \eqref{pe22} and add the upper bound restriction $y_{ij}\leq 1$ for all $i,j$. Let GW($*,\gamma,\delta$) be the resulting model.
\begin{theorem}
GW($*,\gamma,\delta$) is a valid MILP model for QUBO.
\end{theorem}
\begin{proof}
When $x_i=0$ constraints \eqref{pe21} and \eqref{pe22} guarantees that $y_{ij}=y_{ji}=0$. When $x_i=1$ constraints \eqref{pe21} and \eqref{pe22} are redundant. However,  in this case, from constraints \eqref{ch5-e7} and the upper bound restrictions on $y_{ij}$ we have $1\geq y_{ij}\geq x_j$ and $1\geq y_{ji}\geq x_j$. Thus, if $x_j=1$ we have $y_{ij}=y_{ji}=1$. If $x_j=0$ and $x_i=1$, using constraint  \eqref{pe21} and \eqref{pe22} corresponding to $x_j$, we have $y_{ji}=y_{ij}=0$ and the result follows.
\end{proof}
When $\gamma_{ij}=\delta_{ij}=1$ for all $i,j$, we have the corresponding unit weight case.
 Aggregating only \eqref{ch5-e9} and  replacing by \eqref{pe21} and leaving \eqref{ch5-e8} unaltered is also a valid formulation for QUBO and we denote this by GW($*,\gamma$). Likewise, aggregating only \eqref{ch5-e8} as \eqref{pe22}  while leaving \eqref{ch5-e9} unaltered is another  valid formulation and we denote this by GW($*,\delta$). For each of the above variations, we also have the corresponding unit-weight cases (i.e. $\gamma_{ij}=1$ and $\delta_{ij}=1$ for all $(i,j)$). We are skipping the proofs for these two formulations. The number of general constraints for GW($*,\gamma+\delta$), GW($*,\gamma,\delta$), and  GW($*,\gamma$)(GW($*,\delta$) )   are $n^2, \;n(n+1)\mbox{ and } n(2n-1)$ respectively.


\subsection{Simultaneous weighted aggregations}
Simultaneously aggregating type 1 and type 2 constraints in our basic models also lead to valid MILP formulations of QUBO with significantly reduced number of constraints. This is however not applicable for DW since type 2 linearization constraints cannot be aggregated the way we discussed aggregations here. Let us first examine simultaneous aggregation in the context of GW. Let GW($\alpha,\gamma+\delta$) be the MILP obtained by replacing constraints \eqref{ch5-e7} with \eqref{pe9} and replacing constraints \eqref{ch5-e9} and \eqref{ch5-e8} with \eqref{pe14} and adding the upper bound restrictions $y_{ij}\leq 1$. That is,
 \begin{align}
\nonumber\text{GW($\alpha,\gamma+\delta$):~~ Maximize}\hspace{6pt} &\sum_{(i,j)\in S}q_{ij}y_{ij}+\sum_{i=1}^nc_ix_i \\
\label{p17}\mbox{Subject to: }&\sum_{j\in R_i}\alpha_{ij}(x_j-y_{ij})\leq \left(\sum_{j\in R_i}\alpha_{ij}\right)(1-x_i), \mbox{ for } i=1,2,\ldots ,n.\\
\label{p18}&\!\!\sum_{j\in R_i}\gamma_{ij}y_{ij}+\!\!\sum_{j\in S_i}\delta_{ij}y_{ji}\leq\!\! \left(\sum_{j\in R_i}\!\!\gamma_{ij}+\!\!\sum_{j\in S_i}\!\!\delta_{ij}\right)x_i, \mbox{ for } i\!\!=\!\!1,2,\ldots ,n.\\
\label{p19}&x_i\in \{0,1\} \mbox{ for all } i=1,2\ldots ,n.\\
\label{p20}&0\leq y_{ij} \leq 1 \mbox{ for all } j\in R_i, i=1,2,\ldots ,n.
\end{align}
\begin{theorem}
GW($\alpha,\gamma+\delta$) is a valid MILP model for QUBO.
\end{theorem}
\begin{proof}When $x_i=0$, constraint \eqref{p18} assures that $y_{ij}=0$ for all $j\in R_i$ and $y_{ji}=0$ for all $j\in S_i$ .(Since $Q$ is symmetric, $R_i=S_i, i=1,2,\ldots ,n$.) When $x_i=1$, using arguments similar to those in the proof of Theorem~\ref{th1-pk}, it can be shown that $y_{ij}=x_j$ for all $j\in R_i$ and $y_{ji}=x_j$ for all $j\in S_i$. This establishes the validity of the model.
\end{proof}

GW($\alpha,\gamma+\delta$) have $2n$ general constraints, $n(n-1)$ upper bound constraints,  $n$ binary variables, and $n(n-1)$ continuous variables. This is the most compact among the explicit formulations.

Other variations of GW($\alpha,\gamma+\delta$) are also interesting. For example, we can replace \eqref{p18}  by \eqref{pe21} and \eqref{pe22} and let GW($\alpha,\gamma,\delta$) be the resulting model which is valid for QUBO.  The model GW($\alpha,\gamma$) obtained by replacing \eqref{p18} by \eqref{pe21} and \eqref{ch5-e8} is yet another valid MILP for QUBO. Likewise, the model GW($\alpha,\delta$) obtained by replacing \eqref{p18} by \eqref{pe22} and \eqref{ch5-e9} is also a valid model for QUBO.

Let us now combine the two types of weighted aggregations in FT (i.e. inequalities \eqref{pe9} and \eqref{f9}). The resulting MILP formulation is
\begin{align}
\nonumber\text{FT($\alpha,\gamma,*$):\quad Maximize}\hspace{11pt} &\sum_{(i,j)\in S}q_{ij}y_{ij}+\sum_{i=1}^nc_ix_i \\
\label{f11}\mbox{Subject to: }~~&\sum_{j\in R_i}\alpha_{ij}(x_j-y_{ij})\leq \left(\sum_{j\in R_i}\alpha_{ij}\right)(1-x_i), \mbox{ for } i=1,2,\ldots ,n.\\
\label{f12}&\sum_{j\in R_i}\gamma_{ij}y_{ij}\leq \left(\sum_{j\in R_i}\gamma_{ij}\right)x_i, \mbox{ for } i=1,2,\ldots ,n.\\
\label{f15} &y_{ij} \leq y_{ji} \mbox{ for all } j\in R_i,  i=1,2,\ldots ,n.\\
\label{f13}&x_i\in \{0,1\} \mbox{ for all } i=1,2\ldots ,n.\\
\label{f14}&0\leq y_{ij} \leq 1 \mbox{ for all } j\in R_i, i=1,2,\ldots ,n.
\end{align}
\begin{theorem}
FT($\alpha,\gamma,*$) is a valid MILP formulation of QUBO.\label{thm2}
\end{theorem}
\begin{proof}
When $x_i=0$, constraint \eqref{f12} assures that $y_{ij}=0$ for all $j\in R_i$. By constraints \eqref{f15}, whenever $y_{ij}=0$, we have $y_{ji}=0$. When $x_i=1$, using arguments similar to those in the proof of Theorem~\ref{th1-pk}, it can be shown that $y_{ij}=x_j$ for all $j \in R_i$ and by using \eqref{f15} $y_{ji}=x_j$ for all $j\in R_i$. This establishes the validity of the model.
\end{proof}
FT($\alpha,\gamma,*$) have at most $n(n+1)$ general constraints and $n(n-1)$ upper bound constraints.Using similar arguments, it can be shown that FT($\alpha,*,\theta$) is a valid formulation of QUBO as well.  Let FT($\alpha,\gamma,\theta$) be the MILP obtained from FT($\alpha,\gamma,*$) by replacing \eqref{f15} with  \eqref{fte1x-n2}. From the basic model FT, another possible aggregation can be obtained by using following constraint instead:
\begin{equation}\label{f15-n}
\sum_{j\in R_i}\theta_{ij}(y_{ij}-y_{ji})=0
\end{equation}
It may be noted that FT($\alpha,\gamma,\theta$) is a valid formulation of QUBO if the constraint \eqref{f15} in FT($\alpha,\gamma,*$) is replaced by \eqref{fte1x-n2}. But, using \eqref{f15-n} instead does not give a valid formulation of QUBO. This can be established using example EX4.
 The number of general constraints in FT($\alpha,\gamma,\theta$) is $3n$.

Finally, it is possible to generate a combined aggregation model from PK, PK($\alpha,\beta$). However, such a model is a special case of GW($\alpha,\gamma+\delta$) with $\gamma = \delta$, hence will be a valid formulation of QUBO.

\section{Optimality restricted variations and aggregations}\label{s3}
Let us now discuss the optimality restricted variations of the weighted aggregation based models discussed in section \ref{s2}. This needs to be done with more care. In this case, we are loosing some of the advantage gained by optimality restricted variations discussed earlier while gaining some other advantages.
\subsection{Weighted aggregation of type 1 constraints}
Note that under the optimality restricted variations, the number of type 1 constraints are reduced. Let us consider weighted aggregation of these reduced type 1 linearization constraints.
Let $\alpha_{ij} > 0$ for $j\in R_i^-, i=1,2,\ldots ,n$. Now consider the inequality
\begin{equation}\label{p2x-n}
\sum_{j\in R^-_i}\alpha_{ij}(x_j-y_{ij})\leq \left(\sum_{j\in R^-_i}\alpha_{ij}\right)(1-x_i), \mbox{ for } i=1,2,\ldots ,n.
\end{equation}
obtained by the weighted aggregation of the  type 1 linearization constraints in optimality restricted variations of the basic models.
Unlike the precise basic models, for the optimality restricted variations of the basic models, if we replace the type 1 linearization constraints by \eqref{p2x-n}, the resulting MILP need not be a valid model for QUBO. To see this, consider the MILP model ORDW and replace constraints \eqref{pe1x} by \eqref{p2x-n}. Let ORDW-A be the resulting MILP. Interestingly, ORDW-A is not a valid MILP model for QUBO. To see this, consider the instance of QUBO (EX6) with
$$Q=\begin{pmatrix}0&-2&0\\ -2&0&-1 \\0&-1&0 \\ \end{pmatrix} \mbox{ and } \vb c\tran=(5,5,0)$$
where $\alpha_{ij}=1\; \mbox{for all } j\in R_{i}$ for $i=1,\dots ,n$. An optimal solution to the corresponding ORDW-A is $x_1 = x_2 = 1, \; x_3 = 0,\; y_{12} = y_{23} = 1,\;y_{13}=y_{21}=y_{31}=y_{32}=0$ with objective function value 7 but the optimal objective function value for this QUBO is 6. Here, $ y_{23}\neq x_{2}x_{3}$. Also, $y_{12}\neq y_{21}$ and $y_{32}\neq y_{23}$.
 However, we can still use \eqref{p2x-n} to replace constraints \eqref{pe1x} provided we put back all remaining constraints from the corresponding basic model rather than the optimality restricted model. More precisely, consider the model ORDW($\alpha$)
\begin{align}
\nonumber\text{Maximize}\hspace{11pt} &\sum_{i=1}^n\sum_{j\in R_i}q_{ij}y_{ij}+\sum_{i=1}^nc_ix_i \\
\label{pe1xa1}\mbox{Subject to: }~~&  \sum_{j\in R^-_i}\alpha_{ij}(x_j-y_{ij})\leq \left(\sum_{j\in R^-_i}\alpha_{ij}\right)(1-x_i), \mbox{ for } i=1,2,\ldots ,n.\\
\label{pe2xa}&  2y_{ij}-x_i-x_j\leq 0 \mbox{ for all } j\in R_i, i=1,2,\ldots ,n,\\
\label{pe3xa}&x_i\in \{0,1\} \mbox{ for all } i=1,2\ldots ,n,\\
\label{pe4xa}&y_{ij} \in \{0,1\} \mbox{ for all } j\in R_i,i=1,2,\ldots, n.
\end{align}
\begin{theorem}\label{thm10}
ORDW($\alpha$) is a valid MILP formulation of QUBO.
\end{theorem}
\begin{proof} Let DW-R be the MILP obtained from DW by replacing constraints \eqref{pe1} by constraints \eqref{pe1x}. Based on our discussion on optimality restricted models, it can be verified that DW-R is a valid MILP model for QUBO. Clearly, every feasible solution of DW-R is a feasible solution of ORDW($\alpha$). We now show that every feasible solution of ORDW($\alpha$) is a feasible solution of DW-R. For this, it is sufficient to show that every feasible solution of ORDW($\alpha$) satisfies
\begin{equation}\label{p2x1}
x_i+x_j-y_{ij}\leq 1\mbox{ for all }j\in R^-_i, i=1,2,\ldots ,n.
\end{equation}
If $x_i=0$ then inequality \eqref{p2x1} is clearly satisfied. Suppose $x_i=1$. Then, from \eqref{pe1xa1},
$$\sum_{j\in R^-_i}\alpha_{ij}(x_j-y_{ij})\leq 0.$$
But then, as in the proof of Theorem~\ref{th1-dw}, we can show that $x_j=y_{ij}$ for all $j\in R^-_i$ and hence \eqref{p2x1} is satisfied and the proof follows.
\end{proof}
The difference between ORDW($\alpha$) and DW($\alpha$) is minor. For the former, there are less number of variables in the aggregated constraints and hence this part of the coefficient matrix could sometimes be sparse. Otherwise, aggregation did not achieve much.

It is also interesting to note that even if we replace constraints \eqref{pe1x} in ORDW by the full weighted aggregation  constraint \eqref{pe9}, still we do not get a valid MILP formulation of QUBO. Details on this are omitted.

Similarly, let ORPK($\alpha$), ORFT($\alpha$), and ORGW($\alpha$) be the MILP obtained respectively from PK, FT, and GW  by replacing the type 1 linearization constraints with
\begin{equation}
\label{pe1xa}  \sum_{j\in R^-_i}\alpha_{ij}(x_j-y_{ij})\leq \left(\sum_{j\in R^-_i}\alpha_{ij}\right)(1-x_i), \mbox{ for } i=1,2,\ldots ,n.
\end{equation}
\begin{theorem}
ORPK($\alpha$), ORFT($\alpha$), and ORGW($\alpha$) are valid optimality restricted MILP models for QUBO
\end{theorem}
\begin{proof}
The proof is similar to that of theorem \ref{thm10}, where we create ORPK($\alpha$)-R1, ORFT($\alpha$)-R1 and ORGW($\alpha$)-R2, where instead of using the type 2 constraints and lower bound constraints from the respective basic model instead of optimality restricted model of the same. Further proceed with the same idea as theorem \ref{thm10}.
\end{proof}

We want to emphasize that, as in the case of ORDW($\alpha$), the type 2 linearization constraints used in ORPK($\alpha$), ORFT($\alpha$), and ORGW($\alpha$) are those corresponding to their respective basic models rather than the optimality restricted versions of these basic models as using the latter will lead to invalid MILP. The only optimality restricted part used is inequality \eqref{pe1xa}.
To clarify this observation, let us consider EX4 once again, where $\alpha_{ij}=1\; \mbox{for all } j\in R_{i}$ for $i=1,\dots ,n$. The solution is
$x_{1}=x_{3}=1,\; x_{2}=0$,\; $y_{12}=y_{31}=1$, rest of the $y_{ij}\mbox{'s are } 0$ with objective function value 35. But the optimal value of this problem is 34 with $x_{1}=x_{3}=1$ and $x_{2}=0$.
\subsection{Weighted aggregation of type 2 constraints}
Recall that we cannot aggregate type 2 linearization constraints from DW to yield a valid MILP formulation for QUBO. This negative result carries over to ORDW as well. Let us now examine weighted aggregation of type 2 linearization constraints of ORPK. That is,
\begin{align}
\nonumber\text{ORPK($*,\beta$):\quad Maximize}\hspace{11pt} &\sum_{i=1}^n\sum_{j\in R_i}q_{ij}y_{ij}+\sum_{i=1}^nc_ix_i \\
\label{fpe1xa}\mbox{Subject to: }& x_i+x_j -y_{ij}\leq 1 \mbox{ for all } j\in R^-_i, i=1,2,\ldots ,n,\\
\label{fpe2xa}&\sum_{j\in R^+_i}\beta_{ij}(y_{ij}+y_{ji})\leq 2\left(\sum_{j\in R^+_i}\beta_{ij}\right)x_i, \mbox{ for } i=1,2,\ldots ,n.\\
\label{fpe3xa}&x_i\in \{0,1\} \mbox{ for all } i=1,2\ldots ,n,\\
\label{fpe4xa}&y_{ij} \geq 0 \mbox{ for all } j\in R_i,i=1,2,\ldots, n.\\
\label{fpe5xa}&y_{ij} \leq 1 \mbox{ for all } j\in R_i^+,i=1,2,\ldots, n.
\end{align}
Notice that in the model above the lower bound constraints $y_{ij} \geq 0 \mbox{ for all } j\in R_i,i=1,2,\ldots, n$ are required in full and cannot be used corresponding to the optimality restricted model of PK. As seen in the example EX6, if $$Q=\begin{pmatrix}0&6&1\\ 6&0&-7 \\1&-7&0 \\ \end{pmatrix} \mbox{ and } \vb c\tran=(-7,-2,-15)$$
where $\beta_{ij}=1\; \mbox{for all } j\in R_{i}$ for $i=1,\dots ,n$ and the lower bounds are used as $y_{ij} \geq 0 \mbox{ for all } j\in R_i^-,i=1,2,\ldots, n$. The solution obtained from the corresponding formulation is
$x_{1}=x_{3}=0,\; x_{2}=1$,\; $y_{12}=y_{13}=y_{21}=1$, $y_{31}=-3$ rest of the $y_{ij}\mbox{'s are } 0$ with the optimal objective function value as 8. But the optimal value of this problem is 3 with solution $x_{1}=x_{2}=1$ and $x_{3}=0$.

It may be noted that constraint $y_{ij}\leq 1$ is redundant for PK model. Since PK and ORPK are the valid models for QUBO. The model ORPK-R given below is also valid for QUBO
 \begin{align}
\nonumber\text{ORPK-R:\quad Maximize}\hspace{11pt} &\sum_{i=1}^n\sum_{j\in R_i}q_{ij}y_{ij}+\sum_{i=1}^nc_ix_i \\
\label{pe5a-1}\mbox{Subject to: }~~&  x_i+x_j -y_{ij}\leq 1 \mbox{ for all } j\in R^-_i, i=1,2,\ldots ,n,\\
\label{pe6a-1}& y_{ij}+y_{ji}-2x_i\leq 0 \mbox{ for all } j\in R^+_i, i=1,2,\ldots ,n,\\
\label{pe7a-1}&x_i\in \{0,1\} \mbox{ for all } i=1,2\ldots ,n.\\
\label{pe8a-2}& y_{ij} \geq 0 \mbox{ for all }  j\in R_i,i=1,2,\ldots, n.\\
\label{pe8a-3}& y_{ij} \leq 1 \mbox{ for all }  j\in R^+_i,i=1,2,\ldots, n.
\end{align}
\begin{theorem}\label{thm7}
ORPK($*,\beta$) is a valid optimality restricted MILP formulation of QUBO.
\end{theorem}
\begin{proof}
Clearly, every feasible solution of ORPK-R is a feasible solution of ORPK($*,\beta$). Now, we will prove that every solution of ORPK($*,\beta$) is also satisfies for ORPK-R. It is sufficient to show that the feasible solution of ORPK($*,\beta$) satisfies \eqref{pe6a-1}. If $x_{i}=1$, constraint \eqref{pe6a-1} is clearly satisfied. On the other hand, if $x_{i}=0$, then \eqref{fpe2xa} becomes:
\begin{equation*}
\sum_{j\in R^+_i}\beta_{ij}(y_{ij}+y_{ji})\leq 0, \mbox{ for } i=1,2,\ldots ,n.\\
\end{equation*}
As $\beta_{ij}>0$, if $y_{ij}+y_{ji}=0$, then $y_{ij}=y_{ji}=0\mbox{ for all } j\in R_i^+$. This satisfies \eqref{pe6a-1}.
\end{proof}
The number of general constraints for the model ORPK($*,\beta$ is at most n+$\sum_{i=1}^{n}|R_{i}^{-}|$.

The optimality restricted version of GW($*,\gamma+\delta$) is denoted by ORGW($*,\gamma+\delta$) and it can be stated as
 \begin{align}
\nonumber\text{ Maximize}\hspace{11pt} &\sum_{(i,j)\in S}q_{ij}y_{ij}+\sum_{i=1}^nc_ix_i \\
\label{mp17}\mbox{Subject to: }&x_{j}+x_{i}-y_{ij}\leq 1 \mbox{ for } j\in R^-_i ,i=1,2,\ldots ,n.\\
\label{mp18}&\sum_{j\in R^+_i}\gamma_{ij}y_{ij}+\sum_{j\in S^+_i}\delta_{ij}y_{ji}\leq \left(\sum_{j\in R^+_i}\gamma_{ij}+\sum_{j\in S^+_i}\delta_{ij}\right)x_i, \mbox{ for } i=1,2,\ldots ,n.\\
\label{mp19}&x_i\in \{0,1\} \mbox{ for all } i=1,2\ldots ,n.\\
\label{mp20}&0\leq y_{ij} \leq 1 \mbox{ for all } j\in R_i, i=1,2,\ldots ,n.
\end{align}
\begin{theorem}\label{thm3}
ORGW($*,\gamma+\delta$) is valid optimality restricted MILP formulation of QUBO.
\end{theorem}
\begin{proof}
In ORGW, if we replace \eqref{ch5-e11x} by \eqref{mp20}, we get a new formulation similar to the case of Theorem \ref{thm7}. It can be easily noticed that it is a valid formulation for QUBO. Let us name this as ORGW-R. Clearly, every feasible solution of ORGW-R is a feasible solution of ORGW($*,\gamma+\delta$). Now, we need to show that every feasible solution of  ORGW($*,\gamma+\delta$) is also a feasible for ORGW-R. It is sufficient to show that every feasible solution of ORGW($*,\gamma+\delta$) also satisfies constraints \eqref{ch5-e9x} and \eqref{ch5-e8x}. If $x_i=0$, then constraint \eqref{mp18} gives $y_{ij}=y_{ji}=0$, which satisfies constraints \eqref{ch5-e9x} and \eqref{ch5-e8x}. Also, if $x_i=1$, these constraints are clearly satisfied.\end{proof}
\begin{align}
\label{mp188}&\sum_{j\in R^+_i}\gamma_{ij}y_{ij}\leq \sum_{j\in R^+_i}\gamma_{ij}x_i, \mbox{ for } i=1,2,\ldots ,n.\\
\label{mp189}&\sum_{j\in S^+_i}\delta_{ij}y_{ji}\leq\sum_{j\in S^+_i}\delta_{ij}x_i, \mbox{ for } i=1,2,\ldots ,n.
\end{align}
we also get optimality restricted models ORGW($*,\gamma,\delta$), ORGW($*,\gamma$) and ORGW($*,\delta$) corresponding to GW($*,\gamma,\delta$), GW($*,\gamma$) and GW($*,\delta$). It can be shown that the models ORGW($*,\gamma,\delta$), ORGW($*,\gamma$) and ORGW($*,\delta$)  are valid optimality restricted MILP models for QUBO using similar ideas as Theorem \ref{thm3}. The number of general constraints for the model ORGW($*,\gamma+\delta$) and ORGW($*,\gamma,\delta$) is at most $n+\sum_{i=1}^{n}|R_{i}^{-}|$ and $2n+\sum_{i=1}^{n}|R_{i}^{-}|$ respectively. For ORGW($*,\gamma$) and ORGW($*,\delta$), the number of general constraints is  $n+\sum_{i=1}^{n}\left(|R_{i}^{-}|+|R_{i}^+|\right)$.

The optimality restricted version of FT($*,\gamma,\theta$) is given by
\begin{align}
\nonumber\text{ORFT($*,\gamma,\theta$):\quad Maximize}\hspace{11pt} &\sum_{(i,j)\in S}q_{ij}y_{ij}+\sum_{i=1}^nc_ix_i \\
\label{mp117}\mbox{Subject to: }&x_{j}+x_{i}-y_{ij}\leq 1 \mbox{ for } j\in R^-_i ,i=1,2,\ldots ,n,\\
\label{mp118}&\sum_{j\in R^+_i}\gamma_{ij}y_{ij}\leq \sum_{j\in R^+_i}\gamma_{ij}x_i, \mbox{ for } i=1,2,\ldots ,n,\\
\label{mp121}&\sum_{j\in R^+_i}\theta_{ij}y_{ij}\leq \sum_{j\in R^+_i}\theta_{ij}y_{ji}, i=1,2\ldots ,n\\
\label{mp119}&x_i\in \{0,1\} \mbox{ for all } i=1,2\ldots ,n,\\
\label{mp1212}& y_{ij} \geq 0 \mbox{ for all } j\in R_i^-, i=1,2,\ldots ,n.\\
\label{mp120}& y_{ij} \leq 1 \mbox{ for all } j\in R_i^+, i=1,2,\ldots ,n.
\end{align}
\begin{theorem}\label{thm4}
ORFT($*,\gamma, \delta$) is a valid MILP model for QUBO.
\end{theorem}
\begin{proof}
By adding, the constraints \eqref{mp120}  to ORFT, we get another valid formulation of QUBO, say ORFT-R similar to the case of Theorem \ref{thm7}. Clearly, every feasible solution of ORFT-R is a feasible solution of ORFT($*,\gamma, \delta$). To prove the other side, it is sufficient to show that every feasible solution of ORFT($*,\gamma, \delta$) satisfies constraints \eqref{f5x} and \eqref{f6x}. If $x_i=0$, using constraints \eqref{mp118}, $y_{ij}=0 \;\forall j\in R_{i} $ and from constraints \eqref{mp119} $y_{ji}=0$. This satisfies constraints \eqref{f5x} and \eqref{f6x}.
\end{proof}
 The number of general constraints for the model ORFT($*,\gamma,\theta$) is at most $n+\sum_{i=1}^{n}|R_{i}^{-}|$. We can get the optimality restricted aggregations in ORFT($*,\gamma$)  and ORFT($*,\theta$) corresponding to the models FT($*,\gamma$)  and FT($*,\theta$). Their validity can be proved using the similar ideas as in Theorem \ref{thm4}. For ORFT($*,\gamma$)  and ORFT($*,\theta$), the number of general constraints is  $n+\sum_{i=1}^{n}\left(|R_{i}^{-}|+|R_{i}^+|\right)$.

\subsection{Simultaneous weighted optimality restricted aggregations of type 1 and type 2 constraints}

Let us now examine simultaneous aggregation of type 1 and type 2 linearization constraints in the context of optimality restricted models. Recall that when type 1 linearization constraints are aggregated (either full version or the optimality restricted version) we must use full type 2 linearization constraints. Thus, for the simultaneous aggregations, it is necessary that we use aggregation of full type 2 linearization constraints, along with full bound constraints on $y_{ij}$ variables. Let us start with the optimality restricted version of simultaneous aggregation of PK. Let ORPK($\alpha,\beta$) be:
 \begin{align}
\nonumber\text{ORPK($\alpha,\beta$):\quad Maximize}\hspace{11pt} &\sum_{i=1}^n\sum_{j\in R_i}q_{ij}y_{ij}+\sum_{i=1}^nc_ix_i \\
\label{fpe1xb}\mbox{Subject to: }& \sum_{j\in R_i^-}\alpha_{ij}(x_j -y_{ij})\leq \sum_{j\in R_i^-}\alpha_{ij}(1-x_i) , i=1,2,\ldots ,n,\\
\label{fpe2xb}&\sum_{j\in R_i}\beta_{ij}(y_{ij}+y_{ji})\leq 2\left(\sum_{j\in R_i}\beta_{ij}\right)x_i, \mbox{ for } i=1,2,\ldots ,n.\\
\label{fpe3xb}&x_i\in \{0,1\} \mbox{ for all } i=1,2\ldots ,n,\\
\label{fpe4xb}&y_{ij} \geq 0 \mbox{ for all } j\in R_i,i=1,2,\ldots, n.\\
\label{fpe5xb}&y_{ij} \leq 1 \mbox{ for all } j\in R_i,i=1,2,\ldots, n.
\end{align}
\begin{theorem}\label{thm9}
ORPK($\alpha,\beta$) is a valid formulation of QUBO.
\end{theorem}
\begin{proof}
Notice that after adding $y_{ij}\leq 1,\forall  j\in R_{i}, i=1,2,\ldots,n$ to ORPK($\alpha,*$), we still get a valid formulation for QUBO. Let us call it  ORPK($\alpha,*$)-R. We need to show that every feasible solution of ORPK($\alpha,\beta$) is a solution of  ORPK($\alpha,*$)-R. For this, we prove that every feasible solution of ORPK($\alpha,\beta$) satisfies
$y_{ij}+y_{ji}\leq 2x_i$, $\forall j\in R_{i},\; i=1,\ldots, n$. Clearly, if $x_i=0,\; y_{ij}=y_{ji}=0$ using constraints \eqref{fpe2xb}. The required constraint set is satisfied. On the other hand, if $x_{i}=1$, then \eqref{fpe2xb} combined with \eqref {fpe5xb} assures that the required constraint is satisfied. Other side of the proof is pretty obvious hence the result follows.
\end{proof}
Similar to the above model, we have optimality restricted aggregation of GW($\alpha,\gamma+\delta$) which is given as under:
 \begin{align}
\nonumber\text{ ORGW($\alpha,\gamma+\delta$):\quad Maximize}\hspace{11pt} &\sum_{(i,j)\in S}q_{ij}y_{ij}+\sum_{i=1}^nc_ix_i \\
\label{mp171}\mbox{Subject to: }& \sum_{j\in R_i^-}\alpha_{ij}(x_j -y_{ij})\leq \sum_{j\in R_i^-}\alpha_{ij}(1-x_i), i=1,2,\ldots ,n,\\\label{mp181}&\sum_{j\in R_i}\gamma_{ij}y_{ij}+\sum_{j\in S_i}\delta_{ij}y_{ji}\leq \left(\sum_{j\in R_i}\gamma_{ij}+\sum_{j\in S_i}\delta_{ij}\right)x_i, \mbox{ for } i=1,2,\ldots ,n.\\
\label{mp191}&x_i\in \{0,1\} \mbox{ for all } i=1,2\ldots ,n.\\
\label{mp201}&0\leq y_{ij} \leq 1 \mbox{ for all } j\in R_i, i=1,2,\ldots ,n.
\end{align}
\begin{theorem}
ORGW($\alpha,\gamma+\delta$) is a valid formulation for QUBO.
\end{theorem}
\begin{proof}
Similar to the Theorem \ref{thm9}, add $y_{ij}\leq 1,\forall  j\in R_{i}, i=1,2,\ldots,n$ to ORGW($\alpha,*$) and form ORGW($\alpha,*$)-R  Then, every feasible solution of ORGW($\alpha,*$) is a feasible solution for ORGW($\alpha,\gamma+\delta$). Further, we want to show that every feasible solution of
ORGW($\alpha,\gamma+\delta$) satisfies constraints \eqref{ch5-e9} and \eqref{ch5-e8}. If $x_{i}=0, $ $ y_{ij}=y_{ji}=0$ using constraints \eqref{mp181} and \eqref{mp191}. The required constraint set is satisfied. On the other hand, if $x_{i}=1$, then \eqref{mp181} combined with \eqref {mp201} assures that the required constraint is satisfied. Hence, the model is valid.

Other variations of ORGW($\alpha,\gamma+\delta$) are also interesting. As we did for the exact simultaneous models in section 3, we can replace \eqref{mp181}  by \eqref{pe21} and \eqref{pe22} and let ORGW($\alpha,\gamma,\delta$) be the resulting model which is valid for QUBO.  The model GW($\alpha,\gamma,*$) obtained by replacing \eqref{mp181} by \eqref{pe21} and \eqref{ch5-e8} is yet another valid MILP for QUBO. Likewise, the model GW($\alpha,*,\delta$) obtained by replacing \eqref{mp181} by \eqref{pe22} and \eqref{ch5-e9} is also a valid model for QUBO.
\end{proof}
Similar to GW and PK, simultaneous aggregated models can be constructed for FT as well. The formulation ORFT($\alpha,\gamma,\theta$) is given as
 \begin{align}
\nonumber\text{ ORFT($\alpha,\gamma,\delta$):\quad Maximize}\hspace{11pt} &\sum_{(i,j)\in S}q_{ij}y_{ij}+\sum_{i=1}^nc_ix_i \\
\label{mp1711}\mbox{Subject to: }& \sum_{j\in R_i^-}\alpha_{ij}(x_j -y_{ij})\leq \sum_{j\in R_i^-}\alpha_{ij}(1-x_i), i=1,2,\ldots ,n,\\
\label{mp1811}&\sum_{j\in R_i}\gamma_{ij}y_{ij}\leq \sum_{j\in R_i}\gamma_{ij}x_{i}, i=1,2,\ldots ,n,\\
\label{mp1812}&\sum_{j\in R_i}\theta_{ij}y_{ij}\leq \sum_{j\in R_i}\theta_{ij}x_{i}, i=1,2,\ldots ,n,\\
\label{mp1911}&x_i\in \{0,1\} \mbox{ for all } i=1,2\ldots ,n.\\
\label{mp2011}&0\leq y_{ij} \leq 1 \mbox{ for all } j\in R_i, i=1,2,\ldots ,n.
\end{align}
\begin{theorem}
ORFT($\alpha,\gamma,\theta$) is a valid formulation for QUBO.
\end{theorem}
The proof of this theorem is similar to that of Theorem \ref{thm9} and hence, is omitted. ORFT($\alpha,\gamma$) and ORFT($\alpha,\theta$) are two other valid optimality restricted models for QUBO, which are derived along the same lines as we have discussed other models. Details on these are omitted.

\section{Linear Programming Relaxations}

Linear programming relaxations (LP relaxations) of the MILP models we have discussed so far provide  valid upper bounds on the optimal objective function value of QUBO. Such bounds can be effectively used in specially designed enumerative algorithmic paradigms such as branch-and-bound, branch and cut, etc. The linear programming relaxation of GW provides an upper bound for the optimal objective function value of QUBO that matches with the roof duality bound~\cite{ad}. Many researchers have tried to obtain MILP formulations of QUBO with the corresponding LP relaxation bound matches with that of GW~\cite{b2,b3,b1,b7} and almost all of these are compact type linearizations. No explicit linearizations, other than the standard linearization itself is known to have this property. We now show that FT and PK and shares this property along with the weighted aggregation models for appropriately chosen weights.

For any mathematical programming model $P$, its optimal objective function value is denoted by $\obj(P)$. Also, if $P$ is an MILP, its LP relaxation is denoted by $\myoverline{P}$. Two MILP models $P1$ and $P2$ are said to be {\it LP-equivalent} if $\obj(\myoverline{P1})=\obj(\myoverline{P2})$ for all instances of $P1$ and $P2$.   $P1$ is said to be {\it stronger than} $P2$ (for maximization problems), if $\obj(\myoverline{P1})\leq\obj(\myoverline{P2})$ for all instances of $P1$ and $P2$ with strict inequality holds for at least one instance. $P1$ and $P2$ are said to be {\it incomparable} if there exist instances of $P1$ and $P2$ such that $\obj(\myoverline{P1})<\obj(\myoverline{P2})$ and also there exists instances P1 and P2 such that $\obj(\myoverline{P1})>\obj(\myoverline{P2})$.

Let us now analyse the LP relaxations of our  basic models and their optimality restricted counterparts.

\subsection{LP relaxation of basic models}

Our next lemma shows that the models DW, GW, FT, and PK are respectively LP-equivalent to their corresponding optimality restricted forms ORDW, ORGW, ORFT, and ORPK.

\begin{lemma}
$\obj(\myoverline{DW})=\obj(\myoverline{ORDW})$, $\obj(\myoverline{GW})=\obj(\myoverline{ORGW})$, $\obj(\myoverline{FT})=\obj(\myoverline{ORFT})$, and $\obj(\myoverline{PK})=\obj(\myoverline{ORPK})$
\end{lemma}
\begin{theorem}\label{xxtha}
$\obj(\myoverline{GW})=\obj(\myoverline{FT})$
\end{theorem}
\begin{proof}
For any indices $i$ and $j$, the constraint $\eqref{f5} $ and $\eqref{f6}$ of $\overline{FT}$ are given by
\begin{align*}
&y_{ij}\leq x_{i},
y_{ji}\leq x_{j},
y_{ij}\leq y_{ji},\mbox{ and }
y_{ji}\leq y_{ij}.
\end{align*}
Thus $y_{ji}= y_{ij}$ and hence
$y_{ij}\leq x_{i} \mbox{ and } y_{ij}\leq x_{j}$.
Thus, every feasible solution of $\myoverline{FT}$ is feasible for $\myoverline{GW}$ establishing that $\obj(\myoverline{GW})\geq \obj(\myoverline{FT})$

On the other hand, we will show that every optimal solution of $\overline{GW}$ must be also a feasible solution of $\overline{FT}$ Suppose not, then there exist an optimal solution, say $(\overline{x}, \overline{y})$ with objective value $v(\overline{GW}_{(\overline{x},\overline{y})})$, such that at least one constraint of $\overline{FT}$ is violated. As the given solution is a feasible solution of $\overline{GW}$, assume that for some pair $(p,q)$ constraint \eqref{f6} is violated:
\begin{align*}
&\overline{y}_{qp} <\overline{ y}_{pq}
\end{align*}
Since we assume the $Q$ matrix to be symmetric. Without loss of generality, assume that $q_{ij}=q_{ji}>0$. Then, we can construct other  solution $(x,y)$ for $\overline{GW}$ such that $x_{i}=\overline{x}_{i} \; \forall i=1,\dots n$  and
\begin{equation}
y_{ij}=\begin{cases}
\overline{y}_{qp}+\epsilon&\text{if $(i,j)=(q,p)$}\\
\overline{y}_{ij}&\text{otherwise}
\end{cases}
\end{equation}
Clearly, for a small $\epsilon>0$ this is a feasible solution as
\begin{equation}
x_{i}+x_{j}-1\leq y_{qp}=\overline{y}_{qp}+ \epsilon<\overline{ y}_{pq}\leq\text{min}\{{x_{p},x_{q}}\}
\end{equation}
The optimal value of this solution will be $v(\overline{GW}_{(\overline{x},\overline{y})})+q_{ij}\epsilon$ which contradicts the fact that $(\overline{x},\overline{y})$ is an optimal solution of $\overline{GW}$. Thus, every optimal solution of $\overline{GW}$ must be a feasible solution of $\overline{FT}$. Hence $v(\overline{GW})\leq v(\overline{FT})$.
\end{proof}
Punnen, Pandey, and  Friesen~\cite{ppf} studied the effect of LP relaxations of MILP models of QUBO under different equivalent representations of the matrix $Q$. Recall that we assumed that $Q$ is a symmetric matrix. Interestingly, when $Q$ is not necessarily symmetric, FT could provide a stronger LP relaxation. To see this, consider the instance of QUBO (EX7) where
 $$Q=\begin{pmatrix}
0&3\alpha\\
-\alpha&0
\end{pmatrix}\mbox{ and } \vb c\tran=(-\alpha,-\alpha)$$ where $\alpha \geq 1$. Then, $\myoverline{FT}$ has an optimal solution $x_{1}=x_{2}=y_{12}=y_{21}=0$ with objective function value 0, whereas    $x_{1}=x_{2}=y_{12}=0.5,\;y_{21}=0$ is an optimal solution to $\myoverline{GW}$ with objective function value $\alpha/2.$ This shows that  $\obj(\myoverline{GW})$ could be arbitrarily bad compared to $\obj(\myoverline{FT})$.
We now show that $GW$ is stronger than $DW$.
\begin{theorem}
$\obj(\myoverline{GW})\leq\obj(\myoverline{DW})$. Further, $\obj(\myoverline{DW})$  could be arbitrarily bad compared to $\obj(\myoverline{GW})$.
\end{theorem}
\begin{proof}
Note that constraint \eqref{pe2} in DW is the sum of the constraints \eqref{ch5-e4} and \eqref{ch5-e5}. Thus, every feasible solution of $\myoverline{GW}$ is also feasible to $\myoverline{DW}$. Thus, $\obj(\myoverline{GW})\leq\obj(\myoverline{DW})$. Now consider the example (EX8) where $$Q=\begin{pmatrix}0&\alpha/2\\ \alpha/2&0\end{pmatrix} \mbox{ and } \vb c\tran=(-\alpha,1), \mbox{ where } \alpha > 0.$$ Then, $x_1=x_2=y_{12}=y_{21}=1$ is an optimal solution to $\myoverline{GW}$ with objective function value $1$ whereas for $\myoverline{DW}$, $x_1=0, x_2=1, y_{12}=y_{21}=0.5$ is an optimal solution with objective function value $1+\alpha$.
\end{proof}
Recall that any solution to $\myoverline{PK}$ can be represented by the ordered pair $(\vb x, Y)$ where $\vb x =(x_1,x_2,\ldots ,x_j)$ is a vector and $Y$ is an ordered list consisting of $y_{ij}, i=1,2,\ldots n, j\in R_i$.
\begin{lemma}\label{lm10}
Let $(\bar{\vb x},\bar{ Y})$ be an optimal solution to $\myoverline{PK}$. If $\bar{y}_{uv}\neq \bar{y}_{vu}$ for some indices $u$ and $v$ then $\bar{y}_{uv}+\bar{y}_{vu}=2\bar{x}_{u}$ and $\bar{y}_{vu}+\bar{y}_{uv}=2\bar{x}_{v}$.
\end{lemma}
\begin{proof}
Let $S=\{(i,j) : i=1,2,\ldots ,n, j\in R_i,$ and $ (i,j)\neq (u,v),(v,u)\}$. The the objective function of $\myoverline{PK}$ can be written as
$$q_{uv}y_{uv}+q_{vu}y_{vu}+\sum_{(i,j)\in S}q_{ij}y_{ij}+\sum_{j=1}^nc_jx_j.$$
Without loss of generality assume that $\bar{y}_{uv}< \bar{y}_{vu}$. Thus, $\bar{y}_{vu}\neq 0$ and $\bar{x}_u\neq 0, \bar{x}_v\neq 0$. Since $Q$ is symmetric, note that $q_{uv}=q_{vu}\neq 0$. If possible, assume that $\bar{y}_{uv}+\bar{y}_{vu}<2\bar{x}_{u}$. If $q_{uv} > 0$ we can increase the value of $\bar{y}_{uv}$ to $\bar{y}_{uv}+\epsilon$ for arbitrarily small $\epsilon > 0$ to yield a feasible solution to $\myoverline{PK}$ with objective function value larger than that of $(\bar{\vb x},\bar{Y})$, a contradiction. Likewise, if $q_{uv} < 0$ we can decrease the value of $\bar{y}_{vu}$ to $\bar{y}_{vu}-\epsilon$ to yield a feasible solution with objective function value larger than that of $(\bar{\vb x},\bar{Y})$, a contradiction. Thus, $\bar{y}_{uv}+\bar{y}_{vu}=2\bar{x}_{u}$. The equality $\bar{y}_{vu}+\bar{y}_{uv}=2\bar{x}_{v}$ follows by symmetry.
\end{proof}
\begin{lemma}\label{lm10-1}
If $(\bar{\vb x},\bar{Y})$ is an optimal solution to $\myoverline{PK}$ and $\bar{y}_{uv}\neq \bar{y}_{vu}$ for some indices $u$ and $v$ then $\bar{x}_u=\bar{x}_v$. Further,  $(\vb x, Y)$ is also an optimal solution to $\myoverline{PK}$ where $\vb x= \bar{\vb x}$, $y_{uv}=y_{vu}=\bar{x}_u=\bar{x}_v$ and $y_{ij}=\bar{y}_{ij}$ for all $(i,j)\neq (u,v),(v,u)$.
\end{lemma}
\begin{proof}
The proof of the first part of this lemma follows from Lemma~\ref{lm10}. Let us now consider the proof of the second part. Since $\bar{y}_{uv}\neq \bar{y}_{vu}$,
by Lemma ~\ref{lm10}, $\textoverline{y}_{uv}+\textoverline{y}_{vu}=2\textoverline{x}_{u}$ and $\bar{y}_{vu}+\bar{y}_{uv}=2\bar{x}_{v}$. Also, $\bar{x}_u=\bar{x}_v$.
Thus, it can be verified that  $(\vb{x},Y)$ and $(\bar{\vb x},\bar{Y})$ have the same objective function value and $(\vb{x},Y)$ is a feasible solution to $\myoverline{PK}$. Thus $(\vb{x},Y)$ is  an optimal solution to $\myoverline{PK}$.
\end{proof}
 The proof of Lemma~\ref{lm10-1} implicitly used the fact that $Q$ is symmetric.
\begin{theorem}\label{q99}
$\obj(\myoverline{GW})=\obj(\myoverline{PK})$. 
\end{theorem}
\begin{proof}
The constraint \eqref{pe6} in $\myoverline{PK}$ is the sum of the constraints \eqref{ch5-e9} and \eqref{ch5-e8} in $\myoverline{GW}$. Thus every feasible solution $\myoverline{GW}$ is also feasible to $\myoverline{PK}$ and hence $\obj(\myoverline{GW})\leq \obj(\myoverline{PK})$. Now, let  $(\bar{\vb x},\bar{Y})$ be an optimal solution to $\myoverline{PK}$. If $(\bar{\vb x},\bar{Y})$ is a feasible solution to $\myoverline{GW}$, the proof is over. So, suppose that $(\bar{\vb x},\bar{Y})$ is not a feasible solution to $\myoverline{GW}$.   Let $T=\{(i,j) : i=1,2,\ldots, j\in R_i\}$ such that the  solution $(\bar{\vb x},\bar{Y})$ violates at least one of the constraints pair
\begin{align}
\label{ta11}&y_{ij}\leq x_{i}\\
\label{ta12}&y_{ji}\leq x_{i}
\end{align}
for all $(i,j)\in T$. Note that for any pair $(i,j)\in T$ both of these constraints cannot be violated simultaneously. Thus, precisely one of these constraints is violated.  By assumption, $T\neq \emptyset$. Let $(u,v)\in T$ and hence $y_{uv}\neq y_{vu}$.
By Lemma ~\ref{lm10-1}, we can construct an alternative optimal solution $(\vb x, Y)$ to  $\myoverline{PK}$ such that constraints~\eqref{ta11} and \eqref{ta12} are satisfied for all $(i,j)\notin T-\{(u,v)\}$.
Further, $(\vb x, Y)$ is also an optimal solution to $\myoverline{PK}$ with the property that the number of constraints violated in $\myoverline{GW}$ is reduced by one.  Repeating this process, we can construct a feasible solution in $\myoverline{GW}$ which is an optimal solution to $\myoverline{PK}$. Thus $\obj(\myoverline{GW})\geq \obj(\myoverline{PK})$ and the result follows.
\end{proof}
Theorem ~\ref{q99} need not be true if the underlying matrix $Q$ is not symmetric. . This can be illustrated using the same example as the one constructed for a corresponding result for FT. Here $\myoverline{PK}$ gives an optimal solution $x_{1}=x_{2}=0.5,\; y_{12}=0,\; y_{21}=1$, with objective function value $2\alpha$. Thus, $\myoverline{GW}$ gives tighter bound than $\myoverline{PK}$ in this case.
\subsection{LP relaxations of weighted aggregation based models}

 In the preceding section, we have shown that the LP relaxations of GW, FT, and PK yield the same objective function value. Weighted aggregation models are expected to produce weaker bounds. We now show that careful choices of the multipliers used in some of our weighted aggregation models can provide a LP relaxation bound matching the LP relaxation bound of GW. Before discussing these results, let us prove a property from linear programming duality which is related surrogate duality~\cite{dyer, glover2, glover3}.

Consider the linear program
 \begin{align*}
\text{P:\quad Maximize}\hspace{11pt} &\vb c\tran\x \\
\mbox{Subject to: }~~& A\x \leq \vb b, \x\geq \vb 0
\end{align*}
where the matrix $A$ and vectors $\vb b,\vb c\tran$ and $\x$ are  of appropriate dimensions. Let $\vb w^*$ be a given non-negative row vector. Now consider the continuous knapsack problem obtained from $P$ as
 \begin{align*}
\text{CKP($\vb w^*$):\quad Maximize}\hspace{11pt} &\vb c\tran\x  \\
\mbox{Subject to: }~~& \vb w^*A\x \leq \vb w^*\vb b, \x\geq \vb 0
\end{align*}
\begin{theorem}\label{th2}
If P has an optimal solution with an optimal dual solution $\vb w^0$, then CKP($\vb w^0$) also has an optimal solution. Further, the optimal objective function values of P and CKP($\vb w^0$) are the same.
\end{theorem}
Now, consider the linear program
 \begin{align*}
\text{$P^*$:\quad Maximize}\hspace{11pt} &\vb c\tran\x  \\
\mbox{Subject to: }~~& A\x \leq \vb b, A^1\x\leq \vb b^1,A^2\x\leq \vb b^2,\ldots ,A^p\x\leq \vb b^p,  \x\geq \vb 0
\end{align*}
where $\vb b^i\in R^{m_i}$ for $i=1,2,\ldots ,p$. Let $\vb w^i\in R^{m_i}$ be a given non-negative row vectors in $R^{m_i}$ , for $i=1,2,\ldots ,p$. Now consider the new linear program $P^w$ obtained from $P^*$ using weighted aggregation of constraints $A^i\x\leq \vb b^i, i=1,2,\ldots ,n$. Then $P^w$ can be written as
 \begin{align*}
\text{$P^w$:\quad Maximize}\hspace{11pt} &\vb c\tran\x  \\
\mbox{Subject to: }~~& A\x \leq \vb b, \vb w^1A^1\x\leq \vb w^1\vb b^1,\vb w^2A^2\x\leq \vb w^2\vb b^2,\ldots ,\vb w^pA^p\x\leq \vb w^p\vb b^p,  \x\geq \vb 0
\end{align*}
\begin{theorem}\label{th3}
When $\vb w^i$ is the part of an optimal dual solution $\vb w$ of $P^*$ that is associated with the constraint block $A^i\x\leq \vb b^i, i=1,2,\dots ,p$, the optimal objective function values of $P^*$ and $P^w$ are the same.
\end{theorem}
The proof of the theorem~\ref{th2} and theorem~\ref{th3} are skipped due to space constraints. The proofs can be found in the upcoming thesis\cite{nk}

Let us now analyze the quality of LP relaxations of the MILP models generated by weighted aggregations.

\begin{theorem}\label{th5}If the multipliers in the weighted aggregation models is used as  corresponding optimal dual variables associated with those constraints then the resulting aggregated model will have the same objective function value as the corresponding original model. In particular,
\begin{enumerate}[label={(\roman*)},itemindent=1em]
\item If $\alpha_{ij}$ in the type 1 aggregation models is selected as the dual variable associated with the corresponding aggregated constraint in the corresponding models (DW,GW,FT, and PK), then $\obj(\myoverline{DW})=\obj(\myoverline{DW}(\alpha))$, $\obj(\myoverline{GW})=\obj(\myoverline{GW}(\alpha))=\obj(\myoverline{FT}(\alpha))=\obj(\myoverline{PK}(\alpha))$.
\item  If $\beta_{ij}$, $\gamma_{ij}$ and $\delta_{ij}$  in the type 2 aggregation models are selected as the dual variable associated with the corresponding aggregated constraint in the corresponding models (GW,FT, and PK), $\obj(\myoverline{GW})=\obj(\myoverline{GW}(*,\gamma))=\obj(\myoverline{GW}(*,\delta))=\obj(\myoverline{GW}(*,\gamma+\delta))=\obj(\myoverline{FT}(*,\gamma))=\obj(\myoverline{FT}(*,\theta))=\obj(\myoverline{PK}(*,\beta))$.
\item  If $\alpha$, $\beta_{ij}$, $\gamma_{ij}$ and $\delta_{ij}$  in the simultaneous aggregation models are selected as the dual variable associated with the corresponding aggregated constraint in the corresponding models (GW,FT, and PK), $\obj(\myoverline{GW})=\obj(\myoverline{GW}(\alpha,\gamma))=\obj(\myoverline{GW}(\alpha,\delta))=\obj(\myoverline{GW}(\alpha,\gamma+\delta))=\obj(\myoverline{FT}(\alpha,\gamma))=\obj(\myoverline{FT}(\alpha,\theta))=\obj(\myoverline{PK}(\alpha,\beta))$.
\end{enumerate}
\end{theorem}
\begin{proof}
The proof of this theorem follows from Theorem~\ref{xxtha}, Theorem~\ref{q99}, and  Theorem~\ref{th3}.
\end{proof}

It may be noted that when the multipliers ($\alpha_{ij}, \beta_{ij}, etc.$) are chosen as the corresponding dual variables, as discussed in Theorem~\ref{th5}, it could invalidate the corresponding  MILP models since these multipliers could be zero for some $i$ and $j$. Note that the validity of these MILP models is guaranteed only when the multipliers are strictly positive. Thus, if any of the multipliers are zero we need to correct it to an appropriately selected value  $\epsilon > 0$ or find an alternative optimal dual solution (if exists) which is not zero. The choice of $\epsilon > 0$ needs to be made by taking into consideration numerical issues associated with the solver used and the tightness of the LP relaxation bound required. In practice, when the optimal dual variables used have value zero, it needs to be adjusted to appropriate positive values as discussed earlier, to maintain the validity of the models. Depending on the level of adjustments, the LP relaxation values may deteriorate a little bit.

%
%
%

\section{Experimental analysis}

In addition to the theoretical analysis of our MILP models for QUBO, we have also carried out extensive experimental analysis to assess the relative merits of these models. The objectives of the experiments were two fold. First we wanted we assess the ability of the models to solve the problems optimally, with a given time threshold. Another objective is to assess the heuristic value of the models. That is, given an upper threshold for running time, identify if any of the models consistently produced better solutions. We are not comparing the models with best known specially designed QUBO solvers as it will be an unfair comparison. However, we believe the insights gained from our experiments can be used to introduce further refinements to existing QUBO solvers or design new ones with improved capabilities.

All of our computational experiments were carried out on different PCs with same configurations. That is, with Windows 10 Enterprise 64-bit operating system with Intel(R) Core(TM) i7-3770 3.40GHz processor and 16 GB memory. Gurobi 9.5.1 was used as the MILP solver on Anaconda's spyder IDE  with Gurobipy interface. Parameter for 'Presolve' and 'Cuts' was set to 0 throughout these experiments. This is to eliminate bias coming from adding well-known cuts generated based on the Boolean quadric polytope~\cite{{pad}}.  We considered three classes of test instances. This include the well-known Beasley instances\cite{biqmac} and Billionet and Elloumi instances\cite{biqmac}, in addition to instances which we generated which we call balanced data set. To limit the experimental runs, we did not test some simultaneous aggregations of FT, that are FT($\alpha,\gamma$), FT($\alpha,\delta$), ORFT($\alpha,\gamma$) and ORFT($\alpha,\delta$) since they have similar behaviour as GW models.

We used the notational convention MU($\alpha, \beta$) to represent unweighted aggregation of the weighted model M($\alpha, \beta$) where M represents the model and $\alpha $ and $\beta$ represents weights/parameters used for the aggregation. For the experiments, on the weighted versions, the weights are taken as optimal dual variables for the respective constraints in the LP relaxation of the model. All of the optimal dual variables with zero value are replaced by 1. For unweighted versions, all the multipliers are taken as 1.

\subsection{Balanced Data Set}
\begin{figure}[h!]
\centering
\includegraphics[width=16cm, scale=1]{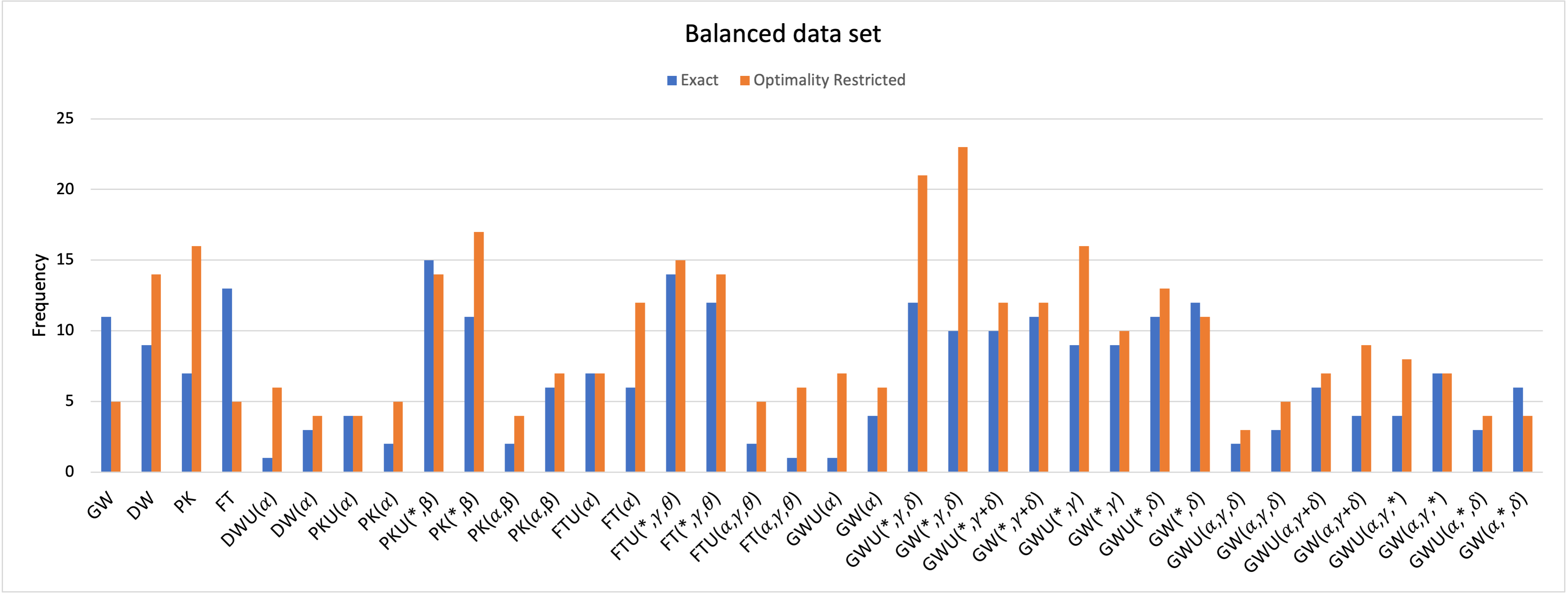}
\caption{Experimental analysis for Balanced data set}
\label{pic1}
\end{figure}
 This data set is generated with the following characteristics: $c_i\in [-10,10]$ for all i, $q_{ij}\in [-20,20]$ for all $i,j$ with $Q$ symmetric and having diagonal elements zero. Ten different problem sizes $n=10,20,30,\ldots ,90,100$ are considered.
After generating a pair $(Q,\vb c)$, we solved the corresponding $\myoverline{GW}$. If all $x_i$ variables turned out to be half-integral, the generated model is accepted into the data set; otherwise the model is discarded. For each problem size (i.e. $n$) five different instances are constructed. The balanced data set we generated will be made available for researchers for future experimental analysis. The time limit was set to 30 minutes for $n\leq60$ and 1 hour for all other instances in this set.

Figue \ref{pic1} shows the experimental results corresponding to the Balanced Data Set for the models discussed in Section 2,3 and 4. For this data set, it is observed that for the smaller instances ($n$$\leq$50), most of the models reached optimality, however, for $n$$\geq$60, for all the instances time limit was reached and the best objective function value is used for the comparison. It can be noticed that among the basic models discussed in section 2, the optimality restricted version of PK model was the best model.

Overall, the type two aggregations of PK, both unweighted and weighted, performed well here.
In general, it can be noted that all the models where type 1 constraints are aggregated never worked better as compared to the aggregations based on Type 2 constraints. For the aggregations of GW, particularly, exact type 2 aggregations performed better than the others.
Hansen and Meyer~\cite{b7}) indicated that replacing constraints \eqref{ch5-e9} and  \eqref{ch5-e8} in GW  by
   \begin{align}\label{pe15}
\sum_{j\in R_i}y_{ij}&\leq |R_i|x_i, \mbox{ for } i=1,2,\ldots ,n.\\
\label{pe16}y_{ij} &\leq 1, \mbox{ for } i=1,2,\ldots ,n.
\end{align}
yields a valid formulation for QUBO and stated it as the Glover-Woolsey aggregation. They showed that the LP relaxation of this model yields very weak upper bound and discarded it from their further experiments. We first observe that the model discussed in~\cite{b7} as Glover-Woolsey aggregation is not precisely the Glover-Woolsey aggregation and that model need not produce an optimal solution for QUBO. For example consider the QUBO with $Q =
\begin{pmatrix}
0 & 1 \\
1 & 0
\end{pmatrix}$ and $\vb c\tran =(-2,0) $. An optimal solution to this problem by this formulation is $x_1=0,\;x_2=1,\; y_{12}=0, \;y_{21}=1$ with objective function value 2 where as $x_1=0,\;x_2=1,\; y_{12}=0, \;y_{21}=0$ is a better solution with optimal objective function value 1 as the former does not give an optimal solution to the original QUBO because $x_1x_2\neq y_{21}$. However, the correct Glover-Woolsey aggregation (with amendments from Goldman~\cite{goldman}) in fact performs computationally better, as can be observed from bar chart \ref{pic1}.

Overall, when we compare all models together, optimality restricted version of GW($*,\gamma,\delta$) turned out to be better formulation with weighted version having a frequency of 23 and unweighted version with frequency 21 for this data set. For the basic models, exact and optimality restricted models performed the same for this data set.
While for the aggregations, the optimality restricted models performed better. For the exact models, unweighted aggregations were better but
for the optimality restricted models, weighted aggregations showed better results. Overall, optimality restricted versions of the aggregated and basic models both outperformed other models for this data set. As mentioned earlier, the bar chart clearly shows most of the type 1 constraints based aggregations and some simultaneous aggregations did not perform well in this experiment, so we discarded PK($\alpha,*$), GW($\alpha,*$), GW($\alpha,\gamma,\theta$), GW($\alpha,\delta$), GW($\alpha,\gamma$), FT($\alpha,*$) and their respective unweighted and optimality restricted versions from next experiments.

\subsection{Beasley Data Set\cite{biqmac}}

The data for this experiment is from OR Library\cite{orlib}. More details of data set can be found in Angelika\cite{biqmac} and Beasley\cite{beasley}.
 We selected instances, n=50,100, 250, 500 and all the coefficients are uniformly distributed integers in [-100,100].
The time limit was set to 30 minutes for n=50 and 1 hour for all other instances.

For n=50 and 100, most of the models produced an optimal solution before reaching the time limit and hence comparison is done based on the run times.

For n=250 and 500, optimality is not reached so we look for the models that give the tightest bounds. Hence, we categorize the comparison on basis of run time and best bounds.
\begin{figure}[h!]
\centering
\includegraphics[width=16cm, scale=1]{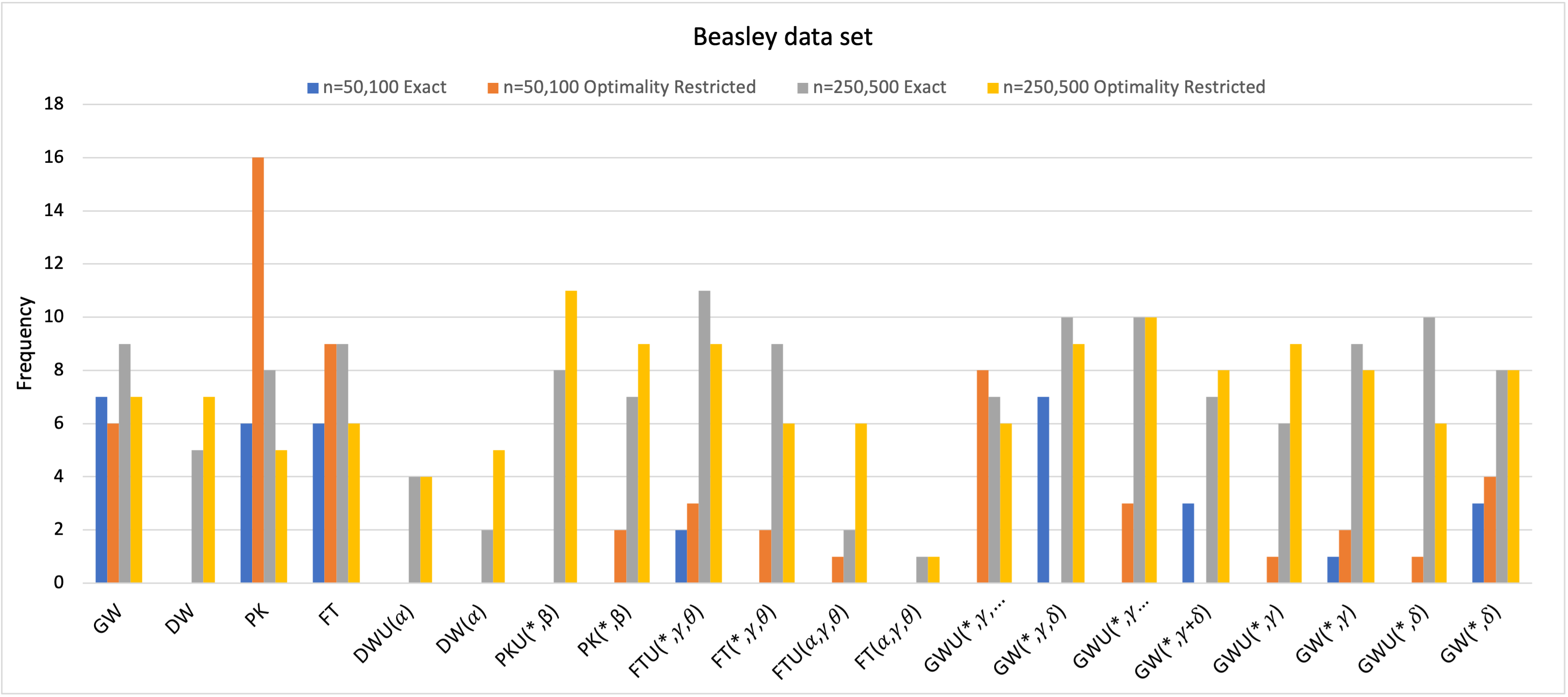}
\caption{Experimental analysis for Beasley data set}
\label{pic2}
\end{figure}

From the figure \ref{pic2} for smaller instances(n=50,100), it is evident that PK model is the best most number of times overall but among
the aggregations optimality restricted version of GW($\beta,\gamma$) was the best. Further, it may be noted that
for this data set, optimality restricted models have better run times.\\
For n=250 and 500 ,unweighted version of the optimality restricted version of PK($*,\beta$) and exact version of FT($*,\gamma,\delta$) gave the
best bound often.

Overall, ORPK and GW($*,\beta,\gamma$) were the best with
frequency 21 and 17  respectively out of 40.
For the basic models, for n=50,100 optimality restricted models
performed better whereas for n=250,500 the exact models were better.
For the aggregations, for n=50,100 optimality restricted models
performed better while for n=250,500 exact and optimality restricted
models were the same.
Also for the aggregations, overall unweighted models were better but
for exact models for n=50,100 weighted models were better.

\subsection{Billionet and Elloumi instances\cite{biqmac}}
\begin{figure}[ht!]
\centering
\includegraphics[width=16cm, scale=1]{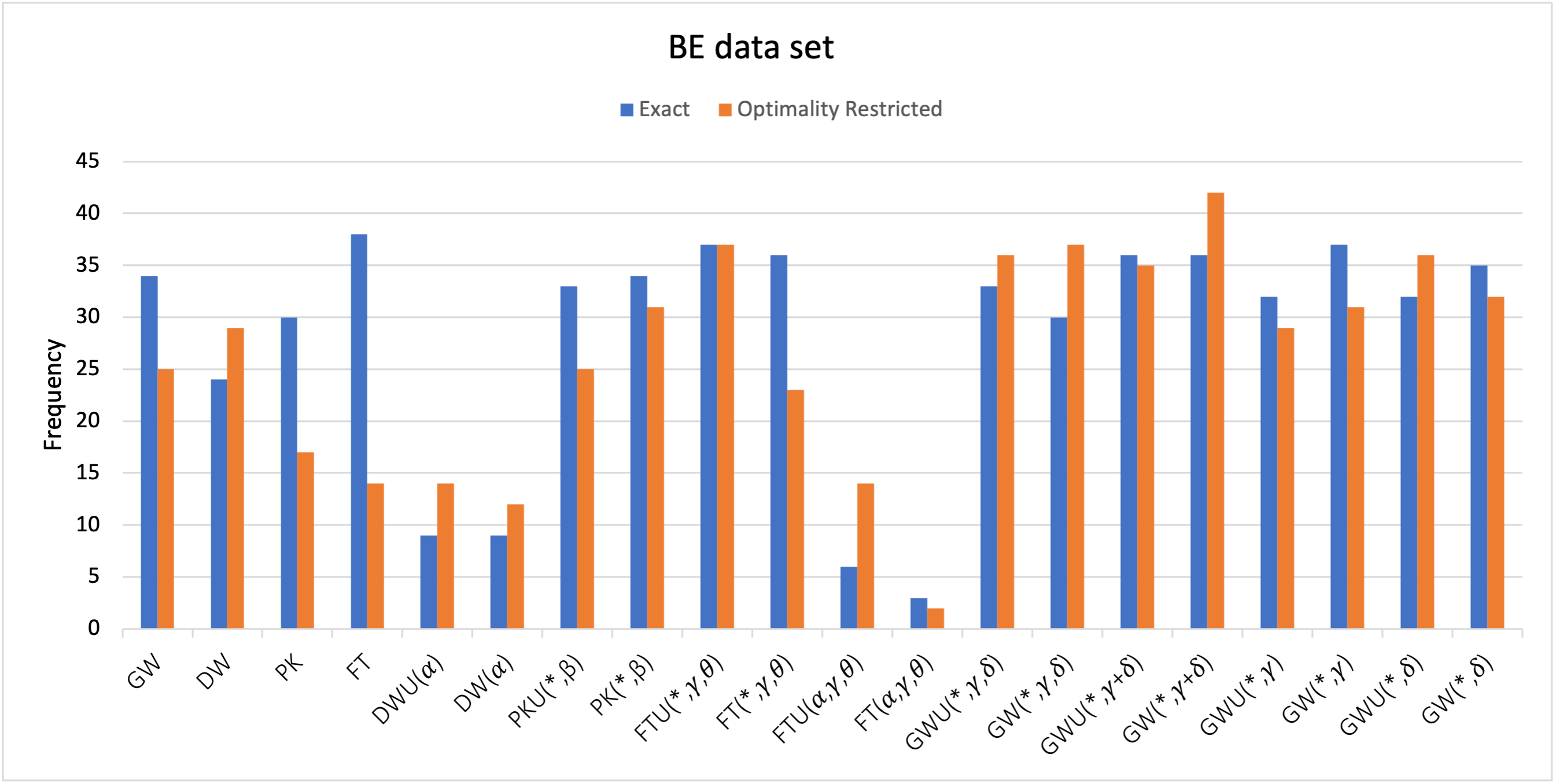}
\caption{Experimental analysis for BE data set}
\label{pic3}
\end{figure}
The third data set we used is Billionet and Elloumi instances~\cite{biqmac,pardalos}. Overall there are 80 instances in this data set with different $n$ and density. We have selected the problems with
n=100,120,150,200,250 where diagonal coefficients are from [-100,100] while off-diagonal
coefficients are from [-50,50]. The density of instances is taken as 1 for n=100, 0.3 and 0.8 for
n=120,150,200 and 0.1 for n=250. The time limit was set 1 hour for all of the instances. For this data set, all the instances reached the time limit. So the comparison is done with respect to value of heuristic solutions.

Overall, from figure \ref{pic3} optimality restricted version of GW($*,\gamma+\delta$) is the best with 42 frequency followed by basic model FT with 38 frequency. It is worth noting that all the aggregations of GW and PK gave good solutions and results are either
better than or at least comparable to GW.
Also, unweighted formulations performed better than weighted formulations for this data set. For the basic models, exact models were better. For the weighted aggregations, exact and optimality restricted models behaved the same but for unweighted aggregations optimality restricted models were better. Also for the aggregations, in general unweighted models were better for this data set.

On considering the analysis for all the data sets it is evident that type 2 aggregations of GW and PK performed better overall.

\section*{Conclusion} In this paper, we have presented various MILP formulations for QBOP.  This includes some basic models and models obtained by  the selective aggregation of constraints of the basic models. Unlike the aggregation based models studied in the literature for general integer programs, our models provide viable alternatives to practical problem solving. Theoretical and experimental analysis on the models has been carried out to assess the relative strength of the model. We also developed  new benchmark instances which will be made available through Github for future experimental work. The applicability of the basic strategies we used in this paper goes beyond QUBO to other quadratic models ( e.g. \cite{meijer1,wu1}) as well as models that are not explicitly quadratic. More details about this will be reported in a sequel. Our computational results disclose that aggregation based models, particularly the type 2 aggregation based models, performed well along with the advantage of having same LP relaxation bound as their corresponding Basic models when choosing the multipliers as optimal dual variables. The aggregation methods suggested in this paper can be used in other integer programming models. For example, the stable set problem, vertex cover problem, facility location, among others. For details on such applications, we refer to the forthcoming Ph.D. thesis~\cite{nk} and our followup papers.

\section*{Acknowledgments} This work was supported by an NSERC discovery grant awarded to Abraham P. Punnen. A preliminary reduced version of this paper was released as a technical report of the mathematics department, Simon Fraser University.


\begin{thebibliography}{99}
\bibitem{ad} W. P. Adams and P. M. Dearing, On the equivalance between roof duality and Lagrangian duality for unconstrained 0-1 quadratic programming problems, {\it Discrete Applied Matheamtics} 48 (1994) 1--20.
\bibitem{b8} W.P. Adams, R.J. Forrester, A simple recipe for concise mixed 0-1 linearizations, {\it Operations Research Letters} 33 (2005) 55–61.
\bibitem{b2} W.P. Adams, R. Forrester and F. Glover, Comparisons and enhancement strategies
for linearizing mixed 0-1 quadratic programs, {\it Discrete Optimization} 1 (2004) 99–120.
\bibitem{adams}  W. P. Adams and  S. M. Henry, Base-2 Expansions for Linearizing Products of
 Functions of Discrete Variables, {\it Operations Research} 60 (2012) 1477--1490.
 \bibitem{b3} W.P. Adams and H.D. Sherali, A tight linearization and an algorithm for 0–1 quadratic
programming problems, {\it Management Science} 32 (1986) 1274–1290.
\bibitem{adams2} W. P. Adams and L Waddell. Linear programming insights into solvable cases of the quadratic assignment problem {\it Discrete Optimization} 14(2014), 46–60.
\bibitem{orlib}John E. Beasley, Or-library. http://people.brunel.ac.uk/~mastjjb/jeb/info.
html, 1990.
\bibitem{beasley}J. E. Beasley, Heuristic algorithms for the unconstrained binary quadratic programming
problem. Technical report, The Management School, Imperial College,
London SW7 2AZ, England, 1998
\bibitem{b1}A. Billionnet, S. Elloumi, and M-C. Plateau, Quadratic 0-1 programming: Tightening linear or quadratic convex reformulation by use of relaxations, {\it RAIRO Operations Research}  42 (2008) 103–121.
\bibitem{qcrb}A. Billionnet, S. Elloumi, M.-C. Plateau, Improving the performance of standard solvers for quadratic 0–1 programs by a tight convex reformulation: the QCR method, {\it Discrete Applied Mathematics} 157 (2009) 1185–1197.
\bibitem{b9} W. Chaovalitwongse, P. Pardalos, O.A. Prokopyev, A new linearization technique for multi-quadratic 0–1 programming problems, {\it Operations
Research Letters} 32 (2004) 517–522
\bibitem{cela2}E. Çela, V.G. Deineko and G.J. Woeginger, Linearizable special cases of the QAP {\it Journal of Combinatorial optimization} 31 (2016) 1269–1279.
\bibitem{dan1} G.B. Dantzig, On the significance of solving linear programming problems with some integer variables, {\it Econometrica} 28 (1960) 30--44.
\bibitem{dyer} M. E. Dyer, Calculating surrogate constraints, {\it Mathematical Programming} 19 (1980) 255--278.
\bibitem{ag1} A. A. Elimam and S. E. Elmaghary, On the reduction method for integer linear programs,  {\it Discrete Applied Mathematics} 12 (1982) 241--260.

\bibitem{fortet} R. Fortet, Applications de l’alg\`{e}bre de boole en recherche op\'{e}rationelle, {\it Revue Francaise
Recherche Op\'{e}rationelle} 4 (1959) 5–36.
\bibitem{c1} R. Fortet, L’alg\`{e}bre de boole et ses applications en recherche op\'{e}rationnelle, {\it Cahiers
du Centre d’Etudes de Recherche Op\'{e}rationnelle} 4 (1960) 17–26.
\bibitem{glover2} F. Glover,  Surrogate Constraints, {\it Operations Research} 16 (1968) 741--749.
\bibitem{glover3} F. Glover,  Surrogate Constraint Duality in Mathematical Programming, {\it Operations
Research}   23 (1975) 434--451.
\bibitem{glover1} F. Glover, Improved linear integer programming formulations of nonlinear integer problems, {\it Management Science} 22 (1975) 455--460.
\bibitem{glover4} F. Glover, Tutorial on surrogate constraints approaches for optimization graphs, Research report, {\it Journal of Heuristics} 9 (2003) 175--227.
\bibitem{ag3} F. Glover and D. A. Babayev, New results for aggregating integer-valued equations, {\it Annals of Operations Research} 58 (1995) 227--242.
\bibitem{gw}   F. Glover and E. Woolsey, Converting the 0–1 polynomial programming problem to a 0–1 linear program, {\it Operations Research} 22 (1974) 180--182.
\bibitem{gw1}   F. Glover and E. Woolsey, Further reduction of zero-one polynomial programming problems to zero-one linear programming problems, {\it Operations Research} 21 (1973) 141--161.
\bibitem{goldman} A. J. Goldman,  Linearization in 0-1 Variables: A Clarification, {\it Operations Research} 31 (1983) 946--947.
\bibitem{gueye}S. Gueye and P. Michelon, A linearization framework for unconstrained
quadratic (0-1) problems, {\it Discrete Applied Mathematics} 157 (2009) 1255--1266.
\bibitem{b6}S. Gueye and P. Michelon, “Miniaturized” Linearizations for Quadratic
0/1 Problems, {\it Annals of Operations Research} 140 (2005) 235–261.
\bibitem{b7} P. Hansen and C. Meyer, Improved compact linearizations for the unconstrained quadratic
0–1 minimization problem, {\it Discrete Applied Mathematics} 157 (2009) 1267--1290.
\bibitem{hus} H. Hu and R. Sotirov, The linearization problem of a binary quadratic problem and its applications, {\it Annals of Operations Research} 307 (2021) 229--249.
\bibitem{kp} S. N. Kabadi and A. P. Punnen, An $O(n^4)$ algorithm for the QAP linearization problem, {\it Mathematics of Operations Research} 36 (2011) 754-761.
\bibitem{nk} N. Kaur, Aggregation properties of 0-1 integer programming problems and connections with binary quadratic programs, Forthcoming PhD. thesis,{\it Simon Fraser University} 2024
\bibitem{ag4} K. E. Kendall and S. Zionts, Solving integer programming problems by aggregating constraints, {\it Operations Research} 25 (1977) 346--351.
\bibitem{survey}G. Kochenberger, J-K. Hao, F. Glover, M. Lewis, Z. Lu, H. Wang and Y. Wang, The Unconstrained Binary Quadratic Programming Problem: A Survey, {\it Journal of Combinatorial Optimization} 28 (2014) 58--81.
\bibitem{adam} A. N. Letchford, The Boolean quadric polytope, in  {\it The Quadratic Unconstrained Binary Optimization Problem:
Theory, Algorithms, and Applications}, A. P. Punnen (ed.), Springer, Switzerland, 2022.
\bibitem{b10}  L Liberti, Compact linearization for binary quadratic problems, {\it 4OR} 5 (3) (2007) 231–245.
\bibitem{mathe}V. Maniezzo, M. A. Boschetti, and T. St\"{u}tzle, {\it Matheuristics: Algorithms and Implementations}, Springer 2021.
\bibitem{ag8}G. B. Mathews, On the partition of numbers, {\it Proceedings of the London Mathematical Society} 28 (1897) 486--490.
\bibitem{meijer1} F. de Meijer and R. Sotirov, SDP-Based Bounds for the Quadratic Cycle Cover Problem via Cutting-Plane Augmented Lagrangian Methods and Reinforcement Learning, {\it INFORMS Journal on Computing} 33 (2021) 262--1276.
\bibitem{murty} K. G. Murty, {\it Linear Programming}, Whiley, 1983.
\bibitem{ag6} D. C. Onyekwelu, Computational viability of a constraint aggregation scheme for integer linear programming problems, {\it Operations Research} 31 (1983) 795--801.
\bibitem{pad} M. Padberg, The boolean quadric polytope: Some characteristics, facets and relatives, {\it Mathematical Programming} 45 (1989) 134--172.
\bibitem{ag7} M. W. Padberg, Equivalent knapsack-type formulations of bounded integer linear programs: An alternative approach, {\it Naval Research Logistics Quarterly} 19 (1973) 699--708.
\bibitem{pardalos}P. M. Pardalos and G. P. Rodgers, Computational aspects of a branch and
bound algorithm for quadratic zero-one programming. {\it Computing} 45 (1990) 131–144.
\bibitem{punnen} A. P. Punnen (ed.), {\it The Quadratic Unconstrained Binary Optimization Problem:
Theory, Algorithms, and Applications}, Springer, Switzerland, 2022.
\bibitem{punnen3}A. P. Punnen,  and S. N.  Kabadi, , A linear time algorithm for the Koopmans–Beckmann QAP linearization and related problems, {\it Discrete Optimization} 10 (2013) 200–209.
\bibitem{pk1} A. P. Punnen and N. Kaur, Revisiting some classical explicit linearizations for the quadratic binary optimization problem, {\it Research Report}, Department of Mathematics, Simon Fraser University, 2022.
\bibitem{pk2} A. P. Punnen and N. Kaur, Revisiting compact linearizations of the quadratic binary
optimization problem, {\it Research Report}, Department of Mathematics, Simon Fraser University, 2022.
\bibitem{punnen4} A. P. Punnen, B. D  Woods and S. N. Kabadi,   A characterization of linearizable instances of the quadratic traveling salesman problem.(2017), {\it arXiv:1708.07217}
\bibitem{ppf}  A. P. Punnen, P. Pandey, and M. Friesen, Representations of  quadratic combinatorial optimization problems: A case study using the quadratic set covering problem, {\it Computers \& Operations Research} 112 (2019) 104769.
\bibitem{pust} A. P. Punnen and R. Sotirov, Mathematical programming models and exact algorithms, in  {\it The Quadratic Unconstrained Binary Optimization Problem:
Theory, Algorithms, and Applications}, A. P. Punnen (ed.), Springer, Switzerland, 2022.
\bibitem{ag5}I. G. Rosenberg, Aggregation of equations in integer programming, {\it Discrete Mathematics} 10 (1974) 325--341.
\bibitem{b4} H.D. Sherali and W.P. Adams, {\it A Reformulation-Linearization Technique for Solving
Discrete and Continuous Nonconvex Problems}. Kluwer Academic Publ., Norwell, MA (1999)
\bibitem{b11}  H. Sherali, W. Adams, A hierarchy of relaxations between the continuous and convexhull representations for zero-one programming problems,
{\it SIAM Journal of Discrete Mathematics} 3  (1990) 411–430.
\bibitem{b12}  H. Sherali, W. Adams, A hierarchy of relaxations and convex hull characterizations for mized-integer zero-one programming problems,
{\it Discrete Applied Mathematics} 52 (1994) 83–106.
\bibitem{watters} L. J. Watters, Reduction of integer polynomial programming problems to zero-one linear programming problems, {\it Operations Research} 15
(1967)

\bibitem{biqmac} A. Wiegele, Biq Mac library – a collection of Max-Cut and quadratic 0-1 programming instances of medium size,  (2007)
\bibitem{wu1} Q. Wu, Y. Wang, and F. Glover, Advanced Tabu Search Algorithms for Bipartite Boolean Quadratic Programs Guided by Strategic Oscillation and Path Relinking, {\it INFORMS Journal on Computing} 32 (2019) 74--89.

\bibitem{zangwill}  W. I. Zangwill, Media selection by decision programming, {\it Journal of Advertising Research} 5 (1965) 30--36.
\bibitem{ag2}N. Zhu and K. Broughan, On aggregating two linear Diophantine equatios, {\it Discrete Applied Mathematics} 82 (1998) 231--246.





\end{thebibliography}
\end{document}